\documentclass[runningheads]{llncs}
\usepackage[T1]{fontenc}
\usepackage{graphicx}
\usepackage{a4wide}
\usepackage{amsmath}

\usepackage{amsthm}
\usepackage{mathtools}
\usepackage{amssymb,amsfonts,latexsym}

\usepackage{tabularx,environ}
\usepackage{enumerate}
\usepackage{bm}
\usepackage{algorithm}
\usepackage{algpseudocode}
\usepackage{paralist}

\usepackage{tikz}
\usepackage{tkz-euclide}
\usetikzlibrary{calc,positioning,shapes}
\usetikzlibrary{positioning}
\usetikzlibrary{patterns}
\usepackage[colorlinks,citecolor=blue,bookmarks=true,linktocpage,hypertexnames=false]{hyperref}
\usepackage[mathlines]{lineno}
\usepackage[capitalise]{cleveref}
\usepackage{complexity}
\usepackage{paralist}

\usepackage{multibib}
\newcites{app}{Additional references in the appendix}

\theoremstyle{definition}
\newtheorem{redrule}{Reduction rule}
\newtheorem{brrule}{Branching rule}
\newcommand{\auxDP}{\mathrm{DP}}
\newcommand{\prb}[1]{\textnormal{\scshape #1}}

\usepackage{thm-restate}
\usepackage{apptools}
\newcommand{\restateref}[1]{\IfAppendix{\hyperref[#1]{$\star$}}{\hyperref[#1*]{$\star$}}}

\newcommand{\prbGen}{\textsc{Parity Domination}}
\newcommand{\prbOD}{\textsc{Odd Domination}}

\newcommand{\prbOinT}{\textsc{$1$-In-$3$-SAT}}
\newcommand{\cw}{\mathrm{cw}}
\newcommand{\pw}{\mathrm{pw}}
\newcommand{\vc}{\mathrm{vc}}
\newcommand{\td}{\mathrm{td}}
\newcommand{\fvs}{\mathrm{fvs}}
\newcommand{\cvd}{\mathrm{cvd}}
\newcommand{\tw}{\mathrm{tw}}

\newcommand{\scno}{\textsc{No}}
\newcommand{\scyes}{\textsc{Yes}}
\newcommand{\redcol}{\textsf{red}}
\newcommand{\bluecol}{\textsf{blue}}

\newcommand{\dpvec}{\mathbf{a}}

\newcommand{\id}[1]{\textsf{\detokenize{#1}}}
\newcommand{\true}{\mathsf{true}}
\newcommand{\false}{\mathsf{false}}

\newcommand{\Vopen}{V_{\mathsf{o}}}
\newcommand{\Vclose}{V_{\mathsf{c}}}
\newcommand{\PSI}{\prb{Partitioned Subgraph Isomorphism}}

\crefname{redrule}{Reduction Rule}{Reduction Rules}
\crefname{brrule}{Branching Rule}{Branching Rules}
\crefname{appendix}{appendix}{appendices}
\Crefname{appendix}{Appendix}{Appendices}

\newcommand{\rev}[1]{#1}
\newcommand{\revn}[1]{#1}
\usepackage{orcidlink}

\begin{document}
\title{Parameterized Complexity of Odd Domination and its Generalization}

\titlerunning{Parameterized Complexity of Odd Domination and its Generalization}
%
\author{
Toranosuke Kokai\inst{1}\orcidlink{0009-0009-3464-1213} \and
Rin Saito\inst{1}\orcidlink{0000-0002-3953-4339} \and
Tatsuhiro Suga\inst{1}\orcidlink{0009-0002-1376-4678}\and
Takahiro Suzuki\inst{1}\orcidlink{0009-0005-8433-3789} \and
Yuma Tamura \inst{1}\orcidlink{0009-0001-5479-7006}
}
\authorrunning{
T. Kokai et al.
}
%
\institute{
Graduate School of Information Sciences, Tohoku University, Sendai, Japan\\
\email{\{kokai.toranosuke.s6, rin.saito, suga.tatsuhiro.p5, takahiro.suzuki.q4\}@dc.tohoku.ac.jp, tamura@tohoku.ac.jp}
}
\maketitle              
\begin{abstract}
In the \textsc{Odd Domination} problem, given a graph $G$ and a positive integer $k$, the task is to determine whether there exists a vertex subset $D$ of $G$ such that the closed neighborhood of each vertex in $G$ contains an odd number of vertices from $D$.
In this paper, we investigate the computational complexity of the problem.
When parameterized by the solution size $k$, we establish $\W[1]$-hardness on some restricted graphs and a sharp boundary between fixed-parameter tractability and $\W[1]$-hardness with respect to the girth of the input graph.
Then, we address the problem when parameterized by several structural graph parameters.
Furthermore, we investigate the parameterized complexity of \textsc{Parity Domination}, which is a generalization of \textsc{Odd Domination}.

\end{abstract}

\section{Introduction} \label{sec:intro}

The \prb{Dominating Set} problem is a fundamental problem in algorithmic theory.
Given a graph $G = (V,E)$ and an integer $k \geq 0$, the task is to decide whether there exists a vertex set $D \subseteq V$ of size at most $k$ such that every vertex of $G$ is contained in the closed neighborhood of some vertex in $D$.
The problem is $\NP$-complete and has been central to the development of parameterized complexity theory~\cite{book:CyganFKLMPPS15}.

One of the generalizations of many domination-type problems is the framework of $(\sigma,\rho)$-dominating sets, introduced in \cite{GenDS:Telle94,GenDS:TelleP93}.
Let $\sigma, \rho \subseteq \mathbb{Z}_{\geq 0}$.
A set $D \subseteq V$ is a $(\sigma,\rho)$-dominating set if
\begin{itemize} 
    \item for every vertex $v \in D$, we have $|N_G(v) \cap D| \in \sigma$, and
    \item for every vertex $v \notin D$, we have $|N_G(v) \cap D| \in \rho$.
\end{itemize}
This framework captures a wide range of classical graph problems, including \prb{Dominating Set}, \prb{Independent Set}, \prb{Perfect Code}, and \prb{Induced Bounded-Degree Subgraph}. 

For many choices of $\sigma$ and $\rho$, the associated optimization problem is $\NP$-hard.
Its classical and parameterized complexity has therefore been studied extensively, typically with respect to the solution size $k$~\cite{GenDS:GolovachKS12,GenDS:MeybodiFMP20} and structural graph parameters~\cite{GenDS:Bui-XuanTV13,GenDS:FockeMINSSW25-1,GenDS:FockeMINSSW25-2,GenDS:JaffkeKST19,GenDS:RooijBR09}.
In particular, when both $\sigma$ and $\rho$ are finite or cofinite (i.e., their complements are finite), the complexity landscape is well understood in many settings~\cite{GenDS:Bui-XuanTV13,GenDS:JaffkeKST19,GenDS:RooijBR09}.

However, several natural problems arise even when $\sigma$ and $\rho$ are neither finite nor cofinite.
A canonical example is the \prb{Odd Domination} problem~\cite{GenDom:GassnerH08,GenDom:Sutner88a}, defined by
$\sigma = \{0,2,4,\ldots\}$ and $\rho = \{1,3,5,\ldots\}$, 
\revn{that is, the sets of even nonnegative integers and positive odd integers, respectively.}
Equivalently, given a graph $G=(V,E)$ and an integer $k \ge 0$, the task is to decide whether there exists a set $D \subseteq V$ of size at most $k$ such that every vertex of $G$ has an odd number of vertices from $D$ in its closed neighborhood.




\paragraph{Our contributions.}
We investigate the classical and parameterized complexity of \prbOD.
(See also \cref{fig:results}.)

We begin with classical complexity of \prbOD\ on restricted graph classes.
Broersma and Li showed that \prbOD\ is $\NP$-complete on bipartite graphs~\cite{OddDom:BroersmaL07}.
We show that $\NP$-completeness holds even under strong structural restrictions:
\prbOD\ is $\NP$-complete on planar bipartite graphs \revn{of maximum degree $3$} with girth at least $g$, for any fixed integer $g \ge 3$.
This result strengthens the known hardness results for \prbOD\ previously established on bipartite graphs~\cite{OddDom:BroersmaL07} and on planar graphs of maximum degree~$3$~\cite{Allcols:BlazejJS26}, and answers an open question posed by Broersma and Li~\cite{OddDom:BroersmaL07}.

We next consider parameterized complexity with respect to the solution size $k$.
On general graphs, \prbOD\ is $\W[1]$-hard when parameterized by $k$~\cite{Allcols:BlazejJS26}.
We strengthen this result by showing that \prbOD\ is $\W[1]$-hard even on bipartite graphs of girth~$4$ and degeneracy~$4$, and also on split graphs.
We then establish a sharp complexity boundary with respect to girth.
While \prbOD\ is $\W[1]$-hard on graphs of girth~$4$, we show that it becomes fixed-parameter tractable on graphs of girth at least~$5$, when parameterized by $k$.

\usetikzlibrary{arrows.meta, positioning, shapes.multipart, calc}

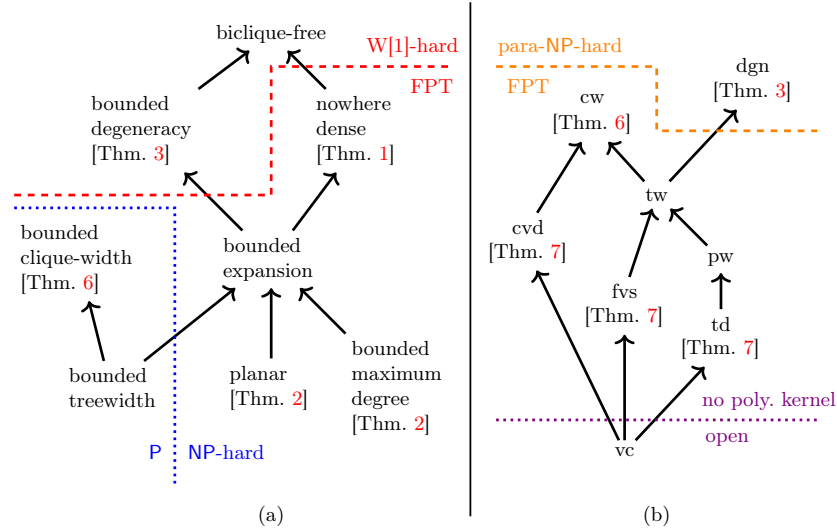
\begin{figure}[t]
\centering
\scalebox{0.85}{
\begin{tikzpicture}

    \node[text width=1.3cm](bounded treewidth) at (0.5, 2){bounded treewidth};
    \node[text width=1.3cm](planar) at (3, 2){planar [Thm. \ref{thm:NPcomp_planar_bipartite}]};
    \node[text width=1.5cm](bounded maximum degree) at (5, 2){bounded maximum degree [Thm. \ref{thm:NPcomp_planar_bipartite}]};
    \node[text width=1.8cm](bounded clique-width) at (0, 4){bounded clique-width [Thm. \ref{thm:cliquew}]};
    \node[text width=1.5cm](bounded expansion) at (3, 4){bounded expansion};
    \node[text width=1.6cm](bounded degeneracy) at (1, 6){bounded degeneracy [Thm. \ref{thm:hard_k}]};
    \node[text width=1.6cm](nowhere dense) at (4.5, 6){nowhere dense [Thm. \ref{thm:FPTnowheredense}]};
    \node[](biclique-free) at (3, 7.5){biclique-free};

    \draw[very thick, ->] (bounded treewidth) -> (bounded clique-width);
    \draw[very thick, ->] (bounded treewidth) -> (bounded expansion);
    \draw[very thick, ->] (planar) -> (bounded expansion);
    \draw[very thick, ->] (bounded maximum degree) -> (bounded expansion);
    \draw[very thick, ->] (bounded expansion) -> (bounded degeneracy);
    \draw[very thick, ->] (bounded expansion) -> (nowhere dense);
    \draw[very thick, ->] (bounded degeneracy) -> (biclique-free);
    \draw[very thick, ->](nowhere dense) -> (biclique-free);

    \draw[very thick, dotted, blue] (-1,4.8) -- (1.5,4.8) -- (1.5,0.5);
    \node[blue] at (1.2,1) {$\P$};
    \node[blue] at (2.3,1) {$\NP$-hard};

    \draw[very thick, dashed, red] (-1,5) -- (3,5) -- (3,7) -- (5.8, 7);
    \node[red] at (5.2,7.3) {W[1]-hard};
    \node[red] at (5.5,6.7) {FPT};

    \node[] at (3,0) {(a)};

    \draw[black, thick] (6.1, 0) -- (6.1,8); 

    \node[] (vc) at (8.5,1){vc};
    \node[] (td) at (10,3){td};
    \node[text width=1.25cm] (td_thm) at (10,2.6){[Thm.~\ref{thm:vc_no_poly_kernel}]};
    \node[] (pw) at (10,4){pw};
    \node[] (tw) at (9,5){tw};
    \node[] (fvs) at (8.5,3.5){fvs};
    \node[text width=1.25cm] (fvs_thm) at (8.5,3.1){[Thm.~\ref{thm:vc_no_poly_kernel}]};
    \node[] (cvd) at (7,4.5){cvd};
    \node[text width=1.25cm] (cvd_thm) at (7,4.1){[Thm.~\ref{thm:vc_no_poly_kernel}]};
    \node[] (cw) at (8,6.5){cw};
    \node[text width=1.25cm] (cw_thm) at (8, 6.1){[Thm.~\ref{thm:cliquew}]};
    \node[] (dgn) at (10.5,7){dgn};
    \node[text width=1.25cm] (dgn_thm) at (10.5, 6.6){[Thm.~\ref{thm:hard_k}]};

    \draw[very thick, ->] (vc) -> (td_thm);
    \draw[very thick, ->] (vc) -> (fvs_thm);
    \draw[very thick, ->] (vc) -> (cvd_thm);
    \draw[very thick, ->] (td) -> (pw);
    \draw[very thick, ->] (fvs) -> (tw);
    \draw[very thick, ->] (pw) -> (tw);
    \draw[very thick, ->] (cvd) -> (cw_thm););
    \draw[very thick, ->] (tw) -> (cw_thm);
    \draw[very thick, ->] (tw) -> (dgn_thm);

    \draw[very thick, dotted, violet] (6.5,1.5) -- (11.5, 1.5);
    \node[violet] at (10.75,1.8){no poly. kernel};
    \node[violet] at (10.1,1.2){open};

    \draw[very thick, dashed, orange] (6.5,7) -- (9,7) -- (9,6) -- (11.5,6);
    \node[orange] at (7.5,7.3){para-$\NP$-hard};
    \node[orange] at (7,6.7){FPT};

    \node[] at (9, 0){(b)};
\end{tikzpicture}
}
\caption{
Summary of our results for \prbOD. 
(a) Relationships among the graph classes considered in this paper.
An arrow $A \to B$ indicates that $A$ is a proper subclass of $B$.
The dotted (blue) line separates polynomial-time solvability from $\NP$-hardness, while the dashed (red) line separates fixed-parameter tractability from $\W[1]$-hardness with respect to the solution size.
(b) Relationships among the graph parameters considered in this paper.
The abbreviations $\mathrm{dgn}$, $\mathrm{cw}$, $\mathrm{tw}$, $\mathrm{cvd}$, $\mathrm{fvs}$, $\mathrm{pw}$, $\mathrm{td}$, and $\mathrm{vc}$ denote the degeneracy, clique-width, treewidth, cluster vertex deletion number, feedback vertex set number, pathwidth, tree-depth, and vertex cover number, respectively.
An arrow $A \to B$ indicates that $A$ is stronger than $B$, that is, if $A$ is bounded by a constant, then so is $B$.
The dashed (orange) line separates fixed-parameter tractability from para-$\NP$-hardness for the corresponding parameters.
The dotted (purple) line indicates that the problem admits no polynomial kernel for parameters above the line, whereas the existence of a polynomial kernel parameterized by $\mathrm{vc}$ remains open.
}
\label{fig:results}
\end{figure}

Our $\W[1]$-hardness result is noteworthy from the viewpoint of sparse \rev{graph classes} \rev{(see again \Cref{fig:results}(a))}.
In contrast, \prb{Dominating Set} is known to be fixed-parameter tractable when parameterized by $k+d$~\cite{DS:AlonG09}, where $d$ denotes the degeneracy.
Moreover, our result complements that \prbOD\ is fixed-parameter tractable on nowhere dense graph classes~\cite{GenDS:GolovachKS12}, which constitute a broad family of sparse graphs and are incomparable with classes of bounded degeneracy~\cite{NesetrilM11a}.

We also investigate structural parameterizations of \prbOD.
For clique-width $\cw$, the previously known approach, derived from results on the problem called \prb{All-Colors}, yields an $O^*(16^{\cw})$-time algorithm for \prbOD~\cite{Allcols:BlazejJS26}.
We substantially reduce the base of the exponent from $16$ to $4$, obtaining an $O^*(4^{\cw})$-time algorithm for \prbOD.

We complement our algorithmic results with kernelization lower bounds.
\rev{Specifically}, we prove that \prbOD\ admits no polynomial kernel when parameterized by the cluster deletion number, tree-depth, or feedback vertex set number (\rev{see \Cref{fig:results}(b)}).
These lower bounds partially resolve an open question raised in~\cite{Allcols:BlazejJS26}.

We also study the complexity of a generalization of \prbOD\ called the \prbGen\ problem, introduced by Gassner and Hatzl~\cite{GenDom:GassnerH08}.
An instance of \prbGen\ consists of a graph $G=(V,E)$, an \emph{offset function} $f \colon V \to \{0,1\}$, a partition $(\Vopen,\Vclose)$ of $V$, and an integer $k \ge 0$.
The question is whether there exists a set $D \subseteq V$ with $|D| \le k$ such that every vertex $v\in V$ satisfies the \emph{parity condition}, defined as follows:
\begin{itemize}
    \item $|N_G(v) \cap D| + f(v) \equiv 0 \pmod{2}$ for every $v \in \Vopen$, and
    \item $|N_G[v] \cap D| + f(v) \equiv 0 \pmod{2}$ for every $v \in \Vclose$.
\end{itemize}

We extend several of our results on \prbOD\ to \prbGen.
In particular, we show that \prbGen\ is fixed-parameter tractable on nowhere dense graph classes when parameterized by $k$, and admits an algorithm running in time $4^{\cw} \cdot n^{O(1)}$.
Furthermore, unless $\coNP \subseteq \NP/\text{poly}$, \prbGen\ admits no polynomial kernel parameterized by the vertex cover number.
This parameter is more restrictive than any parameter considered in our kernelization lower bounds for \prbOD.

\paragraph*{Related work.}
\label{subsec:relwork}
    
It is known that {\prbGen} can be solved in linear time on distance-hereditary graphs and on graphs of bounded treewidth~\cite{GenDom:HatzlW08}, both of which are graph classes of bounded clique-width.

The special case of \prbGen\ with $(\Vopen, \Vclose) = (\emptyset, V(G))$  and a related generalization of \prbOD\ have been studied under the name \textsc{Lights Out!} on graphs from graph-algorithmic and graph-theoretic perspectives~\cite{GenDom:HatzlW08}, as well as from an algebraic perspective~\cite{GenDom:AndersonF98}.
\textsc{Lights Out!} on graphs is derived from the well-known puzzle game of the same name, originally played on a $5 \times 5$ grid.
Each cell has a binary state (ON or OFF), and pressing a button toggles the state of the corresponding cell and all of its neighbors.
\textsc{Lights Out!} can be regarded as a $\sigma$-game~\cite{GenDom:Sutner88a,SigmaGame:Sutner90}, which appears in the context of cellular automata; see the survey by Fleischer et al.~\cite{LightsoutSurvey:FleischerY13} for details.

In the standard \textsc{Lights Out!} setting, each button is either pressed or not pressed, since pressing the same button twice has no effect.
A natural extension is the \prb{All-Colors} problem~\cite{AllCols/ZhangW19a}, where each cell has one of $m$ states and a button may be pressed multiple times. 
Blazej et al.~\cite{Allcols:BlazejJS26} proved that \prb{All-Colors} is $\NP$-complete on planar subcubic graphs established lower bounds based on the Exponential Time Hypothesis (ETH) with respect to the solution size.

Another generalization of \prbOD\ is \prb{Residue Domination}~\cite{ResDom:GreilhuberSW25}, which asks for a $(\sigma,\rho)$-dominating set of size at most $k$ where $\sigma$ and $\rho$ are residue classes modulo $\mathrm{m}\ge 2$.
More precisely, there exist integers $a,b \in \{0,1,\ldots,\mathrm{m}-1\}$ such that
\[
    \sigma = \{x \in \mathbb{Z}_{\geq 0} \mid x \equiv a \pmod{\mathrm{m}}\}
    \quad \text{and} \quad
    \rho = \{x \in \mathbb{Z}_{\geq 0} \mid x \equiv b \pmod{\mathrm{m}}\}.
\]
Greilhuber et al.~\cite{ResDom:GreilhuberSW25} showed that \prb{Residue Domination} can be solved in $\mathrm{m}^{\tw} \cdot n^{O(1)}$ time, where $\tw$ denotes the treewidth of the input graph.
They also proved that, under the Strong Exponential-Time Hypothesis (SETH), the problem admits no $(\mathrm{m}-\varepsilon)^{\pw} \cdot n^{O(1)}$-time algorithm for any $\varepsilon > 0$, where $\pw$ denotes the pathwidth of the input graph.
Consequently, these results imply that \prbOD\ admits an $O^*(2^{\tw})$-time algorithm, but no $(2-\varepsilon)^{\pw} \cdot n^{O(1)}$-time algorithm unless SETH fails.

\paragraph*{Organization.}
In \Cref{sec:preliminaries}, we provide preliminaries and define our problems.
In \Cref{sec:NP-completeness,sec:W1-hardness,sec:FPTgirth5}, we focus on $\prbOD$.
Specifically, we present $\NP$-completeness in \Cref{sec:NP-completeness}, the $\W[1]$-hardness parameterized by $k$ in \Cref{sec:W1-hardness}, and an FPT algorithm parameterized by $k$ for graphs of girth at least $5$ in \Cref{sec:FPTgirth5}.
In \Cref{sec:structual_parameter}, we present an FPT algorithm for $\prbGen$ parameterized by clique-width and kernelization lower bounds with respect to restricted structural parameters.
\revn{The proofs of statements marked with $\star$ are omitted and can be found in the appendix.}

\section{Preliminaries}
\label{sec:preliminaries}
We refer the reader to the book by Cygan et al.~\cite{book:CyganFKLMPPS15} for a detailed introduction to parameterized complexity, and provide only the definitions needed in this paper.

For a positive integer $n$, let $[n] \coloneq \{1,2,\ldots,n\}$, and for integers $l$ and $r$ with $l \le r$, let $[l,r] \coloneq \{l, l+1, \ldots, r\}$.
For an integer $m$, we write $\equiv_m$ for congruence modulo $m$.
For $a \in \{0,1\}$, let $\overline{a} \coloneq 1-a$.

Let $G$ be a graph.
We write $V(G)$ and $E(G)$ for the vertex set and edge set of $G$, respectively.
For a vertex $v \in V(G)$, let $N_G(v)$ and $N_G[v]$ denote the open and closed neighborhoods of $v$, respectively.
For a vertex set $X \subseteq V(G)$, define $N_G(X) \coloneq \{ v \in V(G) \setminus X \mid uv \in E(G), u \in X \}$ and $N_G[X] \coloneq N_G(X) \cup X$.
For $X \subseteq V(G)$, let $G[X]$ denote the subgraph of $G$ induced by $X$.
The \emph{girth} of a graph is the length of its shortest cycle.
A set $I \subseteq V(G)$ is an \emph{independent set} if no two vertices in $I$ are adjacent, and
a set $C \subseteq V(G)$ is a \emph{clique} if every two distinct vertices in $C$ are adjacent.
For a nonnegative integer $d$, a graph $G$ is \emph{$d$-degenerate} if every subgraph of $G$ contains a vertex of degree at most $d$.
The \emph{degeneracy} of $G$ is the minimum such integer $d$.
It is well known that $G$ is $d$-degenerate if and only if it admits an ordering $v_1, v_2, \ldots, v_n$ of $V(G)$ such that each $v_i$ has degree at most $d$ in $G[\{v_i, v_{i+1}, \ldots, v_n\}]$.

Given a graph $G=(V,E)$, an offset function $f \colon V \to \{0,1\}$, and a partition $(\Vopen,\Vclose)$ of $V$, a vertex $v$ is said to be \emph{parity-dominated} by a set $D \subseteq V$ if $v$ satisfies the parity condition.
A set $D$ satisfying the parity condition for every vertex $v \in V$ is called a \emph{parity dominating set} (PD-set) with respect to $f$.
For a given integer $k$, \prbGen\ asks whether there exists a PD-set $D$ of $G$ with size at most $k$.

The problem \prbOD\ is the special case of \prbGen\ obtained by setting $f(v)=1$ for all $v\in V$ and $(\Vopen,\Vclose)=(\emptyset,V)$.
In this case, the parity condition reduces to requiring that $|N_G[v]\cap D|$ is odd for every vertex $v\in V$.
Such a set $D$ is called an \emph{odd dominating set} (ODD-set), and a vertex $v$ is said to be \emph{odd-dominated} by $D$ if $|N_G[v] \cap D|$ is odd.

Following the approach of~\cite{GenDS:GolovachKS12}, one can construct a first-order formula whose length depends only on $k$ and which expresses the existence of a parity dominating set of size at most $k$.
Combined with the first-order model checking meta-theorem of Grohe et al.~\cite{NowheredenseFO:GroheKS17}, this yields the following result.

\begin{restatable}[\restateref{thm:FPTnowheredense}]{theorem}{FPTnowheredense}
\label{thm:FPTnowheredense}
    {\prbGen} admits an FPT algorithm when parameterized by $k$ on nowhere dense graphs. 
\end{restatable}

\section{NP-completeness}
\label{sec:NP-completeness}
In this section, we prove that \prbOD\ remains $\NP$-complete when restricted to planar bipartite graphs of maximum degree $3$ and girth at least $g$, for any fixed integer $g \geq 3$.
To this end, we give a polynomial-time reduction from the \prbOinT\ problem.

Let $\phi$ be a $3$-CNF formula with a variable set $X=\{x_1,x_2,\dots,x_n\}$ and a clause set $C=\{C_1,C_2,\dots,C_m\}$.
\prbOinT\ asks whether there exists a truth assignment $\sigma \colon X \to \{\true,\false\}$ such that each clause of $\phi$ contains exactly one true literal.
The \emph{incidence graph} of $\phi$, denoted by $G_\phi$, is the bipartite graph with vertex set $X \cup C$, where $x_i$ is adjacent to $C_j$ if and only if $x_i$ occurs in $C_j$.
It is known that \prbOinT\ remains $\NP$-complete even if all literals are positive, each variable appears exactly three times, and $G_\phi$ is planar~\cite{1in3SAT:MooreR01}.

We now state the main result of this section.

\begin{theorem}\label{thm:NPcomp_planar_bipartite}
    For every fixed integer $g \ge 3$, \prbOD\ is $\NP$-complete on planar bipartite graphs of maximum degree $3$ and girth at least $g$.
\end{theorem}


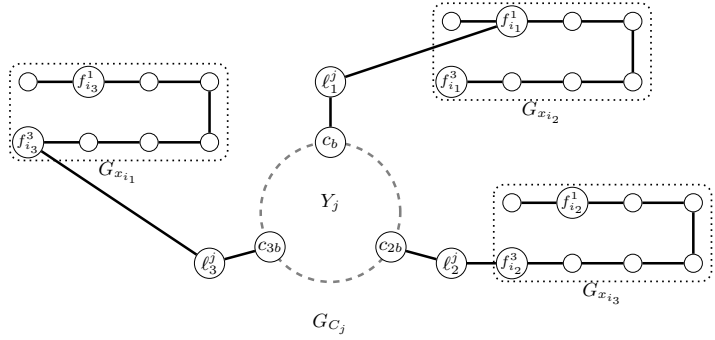
\begin{figure}[t]
    \centering
    \scalebox{0.8}{\begin{tikzpicture}

    \node[] (A) at (5, 7){};
    \node[] (B) at (6, 7-1.73){};
    \node[] (C) at (4, 7-1.73){};

    \tkzDefCircle[circum](A,B,C)
    \tkzGetPoint{O}

    \tkzDrawCircle[very thick, dashed](O,A)

    \node[] at (5, 4) {$G_{C_j}$};
    \node[] at (5, 6) {$Y_j$};
    \node[draw=black, fill=white, circle, minimum size=5mm, inner sep=0pt] (c1) at (5, 7){$c_b$};
    \node[draw=black, fill=white, circle, minimum size=5mm, inner sep=0pt] (c2) at (6, 7-1.73){$c_{2b}$};
    \node[draw=black, fill=white, circle, minimum size=5mm, inner sep=0pt] (c3) at (4, 7-1.73){$c_{3b}$};

    \node[draw=black, fill=white, circle, minimum size=5mm, inner sep=0pt] (l1) at (5, 8){$\ell_1^j$};
    \node[draw=black, fill=white, circle, minimum size=5mm, inner sep=0pt] (l2) at (7, 5){$\ell_2^j$};
    \node[draw=black, fill=white, circle, minimum size=5mm, inner sep=0pt] (l3) at (3, 5){$\ell_3^j$};

    \draw[very thick] (c1)--(l1);
    \draw[very thick] (c2)--(l2);
    \draw[very thick] (c3)--(l3);

    \node[] at (1.5,6.5) {$G_{x_{i_1}}$};
    \draw[thick, dotted, rounded corners=5pt] (-0.3,8.3) rectangle (3.3, 6.7); 
    \node[draw=black, fill=white, circle, minimum size=3mm, inner sep=0pt] (v11) at (7, 9){};
    \node[draw=black, fill=white, circle, minimum size=3mm, inner sep=0pt] (v12) at (8, 9){\scriptsize $f_{i_1}^1$};
    \node[draw=black, fill=white, circle, minimum size=3mm, inner sep=0pt] (v13) at (9, 9){};
    \node[draw=black, fill=white, circle, minimum size=3mm, inner sep=0pt] (v14) at (10, 9){};
    \node[draw=black, fill=white, circle, minimum size=3mm, inner sep=0pt] (v15) at (10, 8){};
    \node[draw=black, fill=white, circle, minimum size=3mm, inner sep=0pt] (v16) at (9, 8){};
    \node[draw=black, fill=white, circle, minimum size=3mm, inner sep=0pt] (v17) at (8, 8){};
    \node[draw=black, fill=white, circle, minimum size=3mm, inner sep=0pt] (v18) at (7, 8){\scriptsize $f_{i_1}^3$};
    \draw[very thick] (v11)--(v12)--(v13)--(v14)--(v15)--(v16)--(v17)--(v18);

    \node[] at (8.5,7.5) {$G_{x_{i_2}}$};
    \draw[thick, dotted, rounded corners=5pt] (6.7,9.3) rectangle (10.3, 7.7); 
    \node[draw=black, fill=white, circle, minimum size=3mm, inner sep=0pt] (v21) at (6+2, 9-3){};
    \node[draw=black, fill=white, circle, minimum size=3mm, inner sep=0pt] (v22) at (7+2, 9-3){\scriptsize $f_{i_2}^1$};
    \node[draw=black, fill=white, circle, minimum size=3mm, inner sep=0pt] (v23) at (8+2, 9-3){};
    \node[draw=black, fill=white, circle, minimum size=3mm, inner sep=0pt] (v24) at (9+2, 9-3){};
    \node[draw=black, fill=white, circle, minimum size=3mm, inner sep=0pt] (v25) at (9+2, 8-3){};
    \node[draw=black, fill=white, circle, minimum size=3mm, inner sep=0pt] (v26) at (8+2, 8-3){};
    \node[draw=black, fill=white, circle, minimum size=3mm, inner sep=0pt] (v27) at (7+2, 8-3){};
    \node[draw=black, fill=white, circle, minimum size=3mm, inner sep=0pt] (v28) at (6+2, 8-3){\scriptsize $f_{i_2}^3$};
    \draw[very thick] (v21)--(v22)--(v23)--(v24)--(v25)--(v26)--(v27)--(v28);

    \node[] at (9.5,4.5) {$G_{x_{i_3}}$};
    \draw[thick, dotted, rounded corners=5pt] (7.7,6.3) rectangle (11.3, 4.7); 
    \node[draw=black, fill=white, circle, minimum size=3mm, inner sep=0pt] (v31) at (6-6, 9-1){};
    \node[draw=black, fill=white, circle, minimum size=3mm, inner sep=0pt] (v32) at (7-6, 9-1){\scriptsize $f_{i_3}^1$};
    \node[draw=black, fill=white, circle, minimum size=3mm, inner sep=0pt] (v33) at (8-6, 9-1){};
    \node[draw=black, fill=white, circle, minimum size=3mm, inner sep=0pt] (v34) at (9-6, 9-1){};
    \node[draw=black, fill=white, circle, minimum size=3mm, inner sep=0pt] (v35) at (9-6, 8-1){};
    \node[draw=black, fill=white, circle, minimum size=3mm, inner sep=0pt] (v36) at (8-6, 8-1){};
    \node[draw=black, fill=white, circle, minimum size=3mm, inner sep=0pt] (v37) at (7-6, 8-1){};
    \node[draw=black, fill=white, circle, minimum size=3mm, inner sep=0pt] (v38) at (6-6, 8-1){\scriptsize $f_{i_3}^3$};
    \draw[very thick] (v31)--(v32)--(v33)--(v34)--(v35)--(v36)--(v37)--(v38);

    \draw[very thick] (v12)--(l1);
    \draw[very thick] (v28)--(l2);
    \draw[very thick] (v38)--(l3);
    
\end{tikzpicture}}
    \caption{An illustration of the gadgets \rev{$G_{x_{i_1}}, G_{x_{i_2}}, G_{x_{i_3}}$}, and $G_{C_j}$ corresponding to the clause $C_j = x_{i_1} \lor x_{i_2} \lor x_{i_3}$.
    Here, $x_{i_1}$ represents the first occurrence of the variable in the formula, while $x_{i_2}$ and $x_{i_3}$ represent its second \rev{and} third occurrences.
    The dashed lines between $c_b, c_{2b}$, and $c_{3b}$ denote \rev{paths of length $b+1$ connecting the corresponding vertices}.}
    \label{fig:NPcomp_gadgets}
\end{figure}

\paragraph{Reduction.}
Let $\phi$ be an instance of \prbOinT\ with \rev{a variable set $X=\{x_1,x_2,\dots,x_n\}$ and a clause set $C=\{C_1,C_2,\dots,C_m\}$.}
We assume that $\phi$ satisfies the above restrictions.
We construct an instance $(G,k)$ of \prbOD.

For each variable $x_i$ with $i\in[n]$, we introduce a \emph{variable gadget} $G_{x_i}$.
The graph $G_{x_i}$ is a path on eight vertices with 
\begin{align*}
    V(G_{x_i})
    &= \{t_i^1,f_i^1,z_i^1,t_i^2,f_i^2,z_i^2,t_i^3,f_i^3\}, \text{ and} \\
    E(G_{x_i})
    &= \{t_i^1f_i^1,f_i^1z_i^1,z_i^1t_i^2,t_i^2f_i^2,f_i^2z_i^2,z_i^2t_i^3,t_i^3f_i^3\}.
\end{align*}

For each clause $C_j = (x^j_1 \lor x^j_2 \lor x^j_3)$ with $j \in [m]$, we define a \emph{clause gadget} $G_{C_j}$.
Let $a$ be the smallest integer such that $g\leq 2^a$, and let $b\coloneqq 2^a$.
The graph $G_{C_j}$ is defined as follows. 
(See also \cref{fig:NPcomp_gadgets})
Construct a cycle $Y_j = (c_1,c_2,\dots,c_{3b})$ of length $3b$ in clockwise order, and add three vertices $\ell^j_1,\ell^j_2,\ell^j_3$ \rev{together with the} edges $\ell_1^jc_b$, $\ell_2^jc_{2b}$, and $\ell_3^jc_{3b}$.

Fix a planar embedding of $G_\phi$.
To construct the graph $G$, we connect the variable gadgets and clause gadgets as follows.
\rev{For each variable vertex $x_i$ and clause vertex $C_j$ of $G_\phi$, replace $x_i$ and $C_j$ with the corresponding gadgets $G_{x_i}$ and $G_{C_j}$, respectively.}
For each variable $x_i$, consider its three occurrences in $\phi$ in an arbitrary but fixed order.
\rev{Suppose that the $t$-th occurrence of $x_i$ ($t \in \{1,2,3\}$) appears as the $r$-th literal of a clause $C_j$.
If $t=1$, add the edge $f_i^1\ell_r^j$; otherwise, add the edge $f_i^3\ell_r^j$.}
Finally, set $k \coloneqq 3n + bm$. 
This completes the construction of the instance $(G,k)$.

We first verify the structural properties of the constructed graph.

\begin{claim}[\rev{$\star$}]
    \label{cl:planar_bipartite_prop}
    The graph $G$ is planar, bipartite, has maximum degree~$3$, and girth at least~$g$.
\end{claim}

The correctness of the reduction is established by the following lemma.

\begin{restatable}[\restateref{thm:NPcplanar}]{lemma}{NPcplanar}
\label{thm:NPcplanar}
    The instance $\phi$ is a yes-instance of {\prbOinT} if and only if the instance $(G,k)$ is a yes-instance of {\prbOD}.
\end{restatable}

\section{$\W[1]$-hardness}\label{sec:W1-hardness}
In this section, we prove that \prbOD\ is $\W[1]$-hard with respect to the solution size~$k$, even on bipartite graphs of degeneracy~$4$ and girth~$4$, and on split graphs.

\subsection{$\W[1]$-hardness on graphs with degeneracy $4$ and girth $4$}
In this subsection, we establish the $\W[1]$-hardness of \prbOD\ parameterized by $k$ even on bipartite graphs with degeneracy~$4$ and girth~$4$.
We also establish a lower bound under the Exponential Time Hypothesis (ETH).

\begin{theorem}
    \label{thm:hard_k}
    \prbOD\ is $\W[1]$-hard when parameterized by $k$ even on bipartite graphs of degeneracy $4$ and girth $4$.
    Moreover, unless ETH fails, there is no algorithm solving \prbOD\ in time $f(k)\cdot n^{o(k/\log k)}$ for any computable function $f$.
\end{theorem}

To prove \Cref{thm:hard_k}, we reduce from \PSI.

An instance of \PSI\ consists of two graphs $H$ and $G$, where $V(H) = [|V(H)|]$ and $V(G)$ is partitioned into color classes $V_1, V_2, \ldots, V_{|V(H)|}$.
The task is to determine whether there exists an injective mapping $\phi \colon V(H) \to V(G)$ such that 
(i) $\phi(i) \in V_i$ for every $i \in V(H)$, and
(ii) $\phi(i)\phi(j) \in E(G)$ for all $ij \in E(H)$.
Such a mapping is called a \emph{colored subgraph isomorphism} from $H$ to $G$.
Without loss of generality, we may assume that every edge of $G$ has endpoints in color classes corresponding to an edge of $H$, since all other edges can be deleted.
In the remainder of this subsection, we assume $i < j$ whenever we write $ij \in E(H)$.

Marx~\cite[Corollary~6.3]{PSI:Marx10} proved that \PSI\ is $\W[1]$-hard when parameterized by $m_H \coloneqq |E(H)|$. 
Moreover, assuming ETH, no algorithm solves an $n$-vertex instance of \PSI\ in time $f(m_H)\cdot n^{o(m_H/\log m_H)}$ for any computable function $f$.


\subsubsection*{Construction.}

Let $(G,H)$ be an instance of \PSI.
Denote $n_H \coloneqq |V(H)|$ and $m_H \coloneqq |E(H)|$.
The vertex set of $G$ is partitioned into $V_1, V_2, \dots, V_{n_H}$.
For each edge $ij \in E(H)$, define $E_{i,j} \coloneq \{ uv \in E(G) \mid u \in V_i, v \in V_j\}$.
By adding isolated vertices and dummy edges if necessary, we may assume that $V_i \neq \emptyset$ for every $i \in V(H)$ and $E_{i,j} \neq \emptyset$ for every $ij \in E(H)$.

We construct an instance $(G',k')$ of \prbOD. 
The graph $G'$ consists of a \emph{vertex gadget} $G_i$ for each $i \in V(H)$ and an \emph{edge gadget} $G_{i,j}$ for each $ij \in E(H)$.
\rev{Finally, let} $k' \coloneqq 4n_H + 4m_H = O(m_H)$.

\paragraph{Vertex gadgets.}
For each $i \in [n_H]$, the vertex gadget $G_i$ consists of the following vertices: a vertex $q_i$; two disjoint copies $V_i^{*}$ and $V_i'$ of $V_i$; a set $L_i \coloneq \{ \ell_{v,i,j} \mid v \in V_i,\ j \in N_H(i) \}$; two vertices $\tilde{x}_i$ and $\tilde{y}_i$; and their respective private neighbors $\tilde{X}_i \coloneq \{\tilde{x}_i^{r} \mid r \in [k'+1]\}$ and $\tilde{Y}_i = \{\tilde{y}_i^{r} \mid r \in [k'+1]\}$.

The edges of $G_i$ are defined as follows.
The vertex $q_i$ is adjacent to every vertex in $V_i^{*}$.
For each $v \in V_i$, let $v^{*} \in V_i^{*}$ and $v' \in V_i'$ denote the corresponding copies; add the edge $v^{*}v'$.
For each $v \in V_i$ and each $j \in N_H(i)$, add the edge $v^{*}\ell_{v,i,j}$.
The vertex $\tilde{x}_i$ is adjacent to every vertex in $V_i^{*}$ and to each vertex \rev{in} $\tilde{X}_i$.
The vertex $\tilde{y}_i$ is adjacent to every vertex in $V_i' \cup L_i \cup \tilde{Y}_i$.
Finally, add the edge $\tilde{x}_i\tilde{y}_i$.

\paragraph{Edge gadgets.}
For each $ij \in E(H)$, the edge gadget $G_{i,j}$ consists of the following vertices: a vertex $q_{i,j}$; two disjoint copies $V_{i,j}^{*}$ and $V_{i,j}'$ of $E_{i,j}$; two vertices $\tilde{x}_{i,j}$ and $\tilde{y}_{i,j}$; and their respective private neighbors $\tilde{X}_{i,j} \coloneq \{\tilde{x}_{i,j}^{r} \mid r \in [k'+1]\}$ and $\tilde{Y}_{i,j} \coloneq \{\tilde{y}_{i,j}^{r} \mid r \in [k'+1]\}$.

For each $uv\in E_{i,j}$ with $u \in V_i$ and $v \in V_j$, let $z^{*}_{uv}\in V_{i,j}^{*}$ and $z'_{uv}\in V_{i,j}'$ denote the corresponding vertices.
The edges of $G_{i,j}$ are defined as follows.
The vertex $q_{i,j}$ is adjacent to all vertices in $V_{i,j}^{*}$.
For each $uv \in E_{i,j}$ with $u \in V_i$ and $v \in V_j$, add the edge $z^{*}_{uv}z'_{uv}$.
The vertex $\tilde{x}_{i,j}$ is adjacent to every vertex in $V_{i,j}^{*}$ and to every vertex in $\tilde{X}_{i,j}$.
The vertex $\tilde{y}_{i,j}$ is adjacent to every vertex in $V_{i,j}' \cup \tilde{Y}_{i,j}$.
Finally, add the edge $\tilde{x}_{i,j}\tilde{y}_{i,j}$.

\medskip
For each edge $uv \in E(G)$ with $u \in V_i$ and $v \in V_j$, add the edges $z^{*}_{uv}\ell_{u,i,j}$ and $z^{*}_{uv}\ell_{v,j,i}$.
This completes the construction of the instance $(G',k')$ of \prbOD.


We verify the structural properties of the constructed graph.

\begin{restatable}[\restateref{lem:w_hard_property}]{lemma}{boundeddg}\label{lem:w_hard_property}
    The graph $G'$ is bipartite, $4$-degenerate, and has girth at least~$4$.
\end{restatable}

We can show that $(G,H)$ is a yes-instance of \PSI\ if and only if $(G',k')$ is a yes-instance of \prbOD.
Moreover, $|V(G')|$ is bounded by a polynomial in $n_H$ and $|V(G)|$, and $k' = 4n_H + 4m_H = O(m_H)$.

Suppose that \prbOD\ can be solved in time $f(k') \cdot |V(G')|^{o(k'/\log k')}$ for some computable function $f$.
Then \PSI\ can be solved in time $f'(m_H) \cdot |V(G)|^{o(m_H/\log m_H)}$ for some computable function $f'$, contradicting ETH.
Hence, no such algorithm for \prbOD\ exists, which completes the proof of \Cref{thm:hard_k}.

\subsection{$\W[1]$-hardness on split graphs.}\label{subsec:split}
A graph is \emph{split} if its vertex set can be partitioned into a clique and an independent set.
We prove the following theorem.
\begin{restatable}[\restateref{thm:hard_split}]{theorem}{thmhardsplit}
    {\prbOD} is $\W[1]$-hard when parameterized by $k$ even for split graphs.
    \label{thm:hard_split}
\end{restatable}

\section{FPT algorithm on graphs with girth at least $5$}\label{sec:FPTgirth5}

In this section, we prove the following theorem.

\begin{theorem}\label{thm:FPT_girth5}
    \prbOD\ can be solved in time $2^{O(k^3)}$ on graphs of girth at least~$5$.
\end{theorem}

To prove \cref{thm:FPT_girth5}, we introduce an annotated variant of \prbGen.
An instance of \prb{Annotated} \prbGen consists of a graph $G=(V,E)$, a partition $(\Vopen,\Vclose)$ of $V$, a set $F \subseteq V$ of \emph{forbidden vertices}, an offset function $f \colon V \to \{0,1\}$, and an integer $k \geq 0$. 
The task is to determine whether there exists a set $D \subseteq V \setminus F$ such that $|D \cup F| \leq k$ and $|N_G[v] \cap D| + f(v) \equiv_2 0 $ for every $v \in V$.

An instance $(G,k)$ of \prbOD\ can be transformed into the equivalent annotated instance $(G,\Vopen,\Vclose,F,f,k)$ by setting $\Vopen = \emptyset$, $\Vclose = V(G)$, $F = \emptyset$, and $f(v)=1$ for all $v \in V(G)$.
Throughout this section, the partition $(\Vopen,\Vclose)$ \rev{is fixed}.
\rev{Accordingly, we abbreviate instances as $(G,F,f,k)$.}

We exhaustively apply the following reduction rules in the order listed below.

\begin{redrule} \label{rule:kzero}
    If $|F| > k$, return ``\scno''.
\end{redrule}

\begin{redrule} \label{rule:solve}
    If $f^{-1}(1) = \emptyset$, return ``\scyes''.
\end{redrule}

\begin{redrule} \label{rule:manyON}
    If there exists a vertex $v$ such that $|N_G(v) \cap f^{-1}(1)| > k$, then return the instance $(G, F', f', k)$, where $F' = F \cup \{v\}$ and
    \[
        f'(u) =
        \begin{cases}
        \overline{f(u)} & \text{if } u \in N_G[v],\\
        f(u) & \text{otherwise}.
        \end{cases}
    \]
\end{redrule}

Each reduction rule can clearly be applied in polynomial time.
Moreover, \Cref{rule:manyON} increases $|F|$ by one while leaving $k$ unchanged.
\revn{Consequently, \Cref{rule:manyON} can be applied at most $k$ times.}
Furthermore, as the graph $G$ is never modified, its girth remains at least~$5$ throughout the reduction process.

The safeness of \cref{rule:kzero,rule:solve} is immediate.
The following lemma is central to our algorithm.

\begin{restatable}[\restateref{lem:girth_degree}]{lemma}{girthdegree}\label{lem:girth_degree}
    After exhaustive application of \Cref{rule:kzero,rule:manyON,rule:solve}, every vertex not in $F$ has degree at most $2k$.
\end{restatable}

We introduce the following branching rule.

\begin{brrule} \label{rule:non_fix_branch}
    Suppose that there exists a vertex $v \in V(G)\setminus F$ such that $f(v)=1$.
    Let $N_G[v]\setminus F = \{v_1,v_2,\dots,v_d\}$.
    For each $i \in [d]$, branch into the instance $(G, F \cup \{v_i\}, f_i, k)$, 
    where
    \[
        f_i(u) =
        \begin{cases}
            \overline{f(u)} & \text{if } u \in N_G[\rev{v_i}],\\
            f(u) & \text{otherwise},
        \end{cases}
    \]
    for every $u \in V(G)$.
\end{brrule}

Each application of \Cref{rule:non_fix_branch} increases $|F|$ by one.
Hence, the depth of the search tree is at most $k$.
By \Cref{lem:girth_degree}, every vertex in $V(G)\setminus F$ has degree at most $2k$.
Therefore, $|N_G[v]\setminus F| \le \deg_G(v)+1 \le 2k+1$. 
It follows that each branching step generates at most $2k+1$ subinstances.
Consequently, the total number of instances generated by \Cref{rule:non_fix_branch} is at most $(2k+1)^k$.

After exhaustive application of \Cref{rule:kzero,rule:solve,rule:manyON,rule:non_fix_branch}, every resulting instance $(G,F,f,k)$ satisfies $|F|\leq k$, $f^{-1}(1)\subseteq F$, and $\deg_G(v)\le 2k$ for every $v\in V(G)\setminus F$.

We next partition the vertices in $V(G) \setminus F$ according to their neighborhoods in $F$.
For each subset $X \subseteq F$, define 
$C_X \coloneq \{\, v \in V(G)\setminus F \mid N_G(v)\cap F = X \}$.
Clearly, the family $\{C_X\}_{X \subseteq F}$ forms a partition of $V(G) \setminus F$.

We iterate over all functions $\mathbf{c} \colon 2^{F} \to [0, k]$ satisfying 
\begin{equation}
    \label{ineq:cdef}
    |F| + \sum_{X \subseteq F} \mathbf{c}(X) \leq k.
\end{equation}
Intuitively, the value $\mathbf{c}(X)$ represents the number of vertices to be selected from $C_X$ \rev{into} the solution.
\revn{The number of such functions is at most the number of ways to assign at most $k-|F|$ selected vertices among the $2^{|F|}$ classes. Hence, it is bounded by $(2^{|F|}+1)^{k - |F|} =2^{O(k^2)}$.}

Fix a function $\mathbf{c}$.
Our goal is to determine whether there exists a set $S \subseteq V(G) \setminus F$ satisfying the following conditions:
\begin{itemize}
    \item $|S \cap C_X| = \mathbf{c}(X)$ for each $X \subseteq F$, and
    \item $S$ is a PD-set for the instance $(G, F, f, k)$.
\end{itemize}
By \Cref{ineq:cdef}, every such set $S$ yields a solution for $(G,F,f,k)$.

To this end, we define the following auxiliary problem. 
An instance consists of a graph $H$ with maximum degree $\Delta(H) \le 2k$, a partition $\mathcal{V} = \{V_1, V_2,\dots, V_{\ell}\}$ of $V(H)$ with $\ell \le 2^k$, and a function $\mathbf{c}' \colon \mathcal{V} \to [0,k]$. 
The task is to determine whether there exists a set $S \subseteq V(H)$ of size at most $k$ such that:
\begin{inparaenum}
    \item[{(1)}] $|N_H(v) \cap S| \equiv_2 0$ for every $v \in V(H)$, and
    \item[{(2)}] $|S \cap V_i| = \mathbf{c}'(V_i)$ for every $i \in [\ell]$.
\end{inparaenum}

We now describe a reduction from our problem to the auxiliary problem. 
For a vertex $v \in F$, define $\mathcal{X}_v \coloneq \{ X \subseteq F \mid v \in X \}$. 
The following reduction rule handles the parity requirements of vertices in $F$.

\begin{redrule} \label{rule:deletionF}
    If there exists a vertex $v \in F$ such that
    \[
        \sum_{X \in \mathcal{X}_v} \mathbf{c}(X) + f(v) \not\equiv_2 0,
    \]
    then return ``\scno''.
\end{redrule}

\rev{Suppose that \Cref{rule:deletionF} does not return ``\scno''.
We construct an instance of the auxiliary problem as follows.
Let $H \coloneq G - F$. 
The partition of $V(H)$ is $\mathcal{V} = \{C_X\}_{X \subseteq F}$.
We further define the function $\mathbf{c}'$ by setting $\mathbf{c}'(C_X) \coloneq \mathbf{c}(X)$ for every $X \subseteq F$.}

Since $|F| \leq k$, we obtain $|\mathcal{V}| \leq 2^k$, which satisfies the requirements of the auxiliary problem. 
Furthermore, by \Cref{lem:girth_degree}, the maximum degree of $H$ is at most $2k$.

\subsubsection{Solving the auxiliary problem}
We provide an $\FPT$ algorithm for the auxiliary problem parameterized by $k$, which completes the proof of \Cref{thm:FPT_girth5}.
Our approach employs the \emph{random separation} technique of Cai et al.~\cite{random_sep:CaiCC06}; see also~\cite[Section~5.3]{book:CyganFKLMPPS15}.

Fix a solution $S \subseteq V(H)$ to the auxiliary problem.
Let $\Lambda_S$ denote the set of edges \rev{in $H$ incident to vertices in $S$.}
Furthermore, let $\Gamma_S$ denote the set of edges \rev{in $H$ incident to vertices in $N_H(S)$ but not incident to vertices in $S$.}

\rev{We independently color each edge of $H$ either $\redcol$ or $\bluecol$ uniformly at random.}
Let $\chi \colon E(H) \to \{\redcol, \bluecol\}$ denote the resulting \rev{edge-coloring}.
We say that $\chi$ is \emph{successful} with respect to $S$ if all edges in $\Lambda_S$ are colored $\redcol$ and all edges in $\Gamma_S$ are colored $\bluecol$.
Since $\Delta(H) \leq 2k$ and $|S| \leq k$, we obtain $|N_H(S)| \leq 2k^2$.
Thus, the number of possible colorings of the edges in $\Lambda_S \cup \Gamma_S$ is bounded by:
\[
    |\Lambda_S| + |\Gamma_S| \leq \Delta(H) \cdot (|S| + |N_{\rev{H}}(S)|) \leq 2k \cdot (k + 2k^2) = 4k^3 + 2k^2.
\]
Therefore, a random coloring $\chi$ is successful with probability at least 
\[
    2^{-(|\Lambda_S|+|\Gamma_S|)}
    \geq 2^{-(4k^3+2k^2)}.
\]

Let $H_R = (V(H), \chi^{-1}(\redcol))$ be the subgraph of $H$ consisting of all red edges. 
We call a vertex $v \in V(H)$ \emph{consistent} if all edges incident to $v$ in $H$ is colored $\redcol$.
Observe that if $\chi$ is successful, every vertex $v \in S$ is consistent, since every edge incident to a vertex in $S$ belongs to $\Lambda_S$.

We design a dynamic programming algorithm to find $S$ from $H_R$.
Let $\{C_i\}_{i = 1}^{t}$ denote the connected components of $H_R$,
and for each $i \in [0,t]$, let $H_R^i$ be the subgraph \rev{of $H_R$} induced by $\bigcup_{j=1}^i V(C_j)$, where $H_R^0$ is the empty graph. 
For each $i \in [0,t]$ and each function $\dpvec \colon [\ell] \to [0,k]$, we define $\auxDP(i, \dpvec) = \true$ if there exists a set $D \subseteq V(H_R^i)$ such that:
\begin{itemize}
    \item $|N_H(v) \cap D| \equiv_2 0$ for every $v \in V(H_R^i)$;
    \item $|D \cap V_j| = \dpvec(j)$ for every $j \in [\ell]$; and
    \item every vertex in $D$ is consistent.
\end{itemize}
Otherwise, define $\auxDP(i, \dpvec) = \false$. 
The base case is given by $\auxDP(0, \dpvec) = \true$ if and only if $\dpvec(j) = 0$ for every $j \in [\ell]$.

For $i \ge 1$, we observe that if $|V(C_i)| > k(2k+1)$, then $C_i$ cannot contain any \rev{vertex} of $S$, since \rev{$|N_H[S]| \leq |S| + |N_H(S)| \leq k + 2k^2 = k(2k+1)$}.
Therefore, in this case, $\auxDP(i, \dpvec) = \auxDP(i-1, \dpvec)$.

\rev{Suppose now that $|V(C_i)| \leq k(2k+1)$.}
Let $\mathcal{D}_i$ denote the family of all subsets $D \subseteq V(C_i)$ of size at most $k$ such that, for all $v \in V(C_i)$, $v$ is consistent and satisfies $|N_H(v) \cap D| \equiv_2 0$. 
Since $|V(C_i)| \leq 2k^2+k$, we obtain \rev{$|\mathcal{D}_i| \leq \sum_{r=0}^{k} \binom{|V(C_i)|}{r} \leq (2k^2+k+1)^k$.}
Moreover, $\mathcal{D}_i$ can be computed in time $O((2k^2+k+1)^k) \cdot n^{O(1)}$.

For each $D \in \mathcal{D}_i$, define the function $\dpvec_D(j) = |D \cap V_j|$ for every $j \in [\ell]$.
We then update the table according to $\auxDP(i, \dpvec) = \bigvee_{D \in \mathcal{D}_i} \auxDP(i-1, \dpvec - \dpvec_D)$.

The number of possible functions $\dpvec$ is at most $2^{O(k^2)}$.
\rev{Hence, each dynamic programming step can be processed in time $(2k^2+k+1)^k \cdot 2^{O(k^2)} \cdot n^{O(1)} = 2^{O(k^2)} \cdot n^{O(1)}$.}

\paragraph{Running time and correctness}

For a fixed coloring $\chi$, the dynamic programming decides in $2^{O(k^2)} \cdot n^{O(1)}$ time whether \rev{there exists a set $S$ satisfying} the requirements \rev{specified by} $\mathbf{c}'$. 
If the given instance is a no-instance, then $\auxDP(t, \mathbf{c}) = \false$ for any coloring $\chi$.
Otherwise, \rev{there exists a solution $S$ to the auxiliary problem.}
\rev{As shown above,} a random coloring $\chi$ is successful with respect to $S$ with probability at least $2^{-(4k^3 + 2k^2)}$, in which case the dynamic programming \rev{algorithm} returns $\true$.
Therefore, the auxiliary problem can be solved by a randomized algorithm in $2^{4k^3+2k^2}\cdot 2^{O(k^2)}\cdot n^{O(1)} = 2^{O(k^3)}\cdot n^{O(1)}$ time with constant probability. 

By \Cref{rule:non_fix_branch}, the algorithm generates at most $(2k + 1)^k$ instances of \prb{Annotated} \prbGen.
For each such instance, we consider at most $2^{O(k^2)}$ candidate functions $\mathbf{c}$.
For every choice of $\mathbf{c}$, we solve the corresponding auxiliary problem in randomized time $2^{O(k^3)} n^{O(1)}$.
Therefore, the overall running time is $(2k+1)^k\cdot 2^{O(k^2)}\cdot 2^{O(k^3)}\cdot n^{O(1)}=2^{O(k^3)}\cdot n^{O(1)}$. 
The algorithm outputs $\true$ only if it finds a valid solution.
Thus, the algorithm is a Monte Carlo algorithm with one-sided error.


Finally, the randomized algorithm can be derandomized using a standard tool called \emph{universal sets}; see, for example, Cygan et al.~\cite{book:CyganFKLMPPS15}.
This completes the proof of \Cref{thm:FPT_girth5}.

\section{Structural Parameters}\label{sec:structual_parameter}

\subsection{Clique-width}\label{subsec:clique-width}

We now state the main result of this subsection.

\begin{restatable}[\restateref{thm:cliquew}]{theorem}{fptcliquew}
    \label{thm:cliquew}
    \prbGen\ is fixed-parameter tractable when parameterized by clique-width.
\end{restatable}

\subsection{No Polynomial Kernel Parameterized by Structural Parameters} 
\label{subsec:kernel}

We use the framework of OR-cross-compositions~\cite{kernel:BodlaenderJK14}, a standard technique for proving kernelization lower bounds.
Our proof follows the approach used to establish kernelization lower bounds for \prb{Dominating Set} parameterized by vertex cover~\cite{book:CyganFKLMPPS15,DS:DomLS14}.

\begin{restatable}[\restateref{thm:vc_no_poly_kernel}]{theorem}{vcnopolykernel}
\label{thm:vc_no_poly_kernel}
Unless $\coNP \subseteq \NP/\text{poly}$,
\prbGen\ admits no polynomial kernel parameterized by vertex cover.
Moreover, \prbOD\ admits no polynomial kernel parameterized by cluster deletion \revn{number}, feedback vertex set \revn{number}, or treedepth.
\end{restatable}

\bibliographystyle{splncs04}
\bibliography{ref}

\appendix

\section{Algorithms by first-order model checking}\label{app:nowheredenseFO}

In this subsection, we prove the following theorem using first-order model checking.
\FPTnowheredense* \label{thm:FPTnowheredense*}

\paragraph{Nowhere dense graphs.}
\revn{A graph $H$ is a \emph{minor} of a graph $G$, denoted by $H \preceq G$, if each vertex $v\in V(H)$ can be associated with a nonempty connected subgraph $G_v$ of $G$ such that the subgraphs $\{G_v\}_{v\in V(H)}$ are pairwise vertex-disjoint, and for every edge $uv\in E(H)$, there exists an edge of $G$ with one endpoint in $G_u$ and the other in $G_v$.
The subgraph $G_v$ is called the \emph{branch set} corresponding to $v$.
For a nonnegative integer $r$, we say that $H$ is a \emph{depth-$r$ minor} of $G$, denoted by $H \preceq_r G$, if there exists such a collection of branch sets $\{G_v\}_{v\in V(H)}$ in which every branch set has radius at most $r$.
A graph class $\mathcal C$ is \emph{nowhere dense} if there exists a function $t\colon\mathbb{N}\to\mathbb{N}$ such that $K_{t(r)} \not\preceq_r G$ for every $r\in\mathbb N$ and every $G\in\mathcal C$.}


\newcommand{\onV}{V_{\mathrm{on}}}
\newcommand{\offV}{V_{\mathrm{off}}}

\paragraph{First-order logic.}
We refer the reader to~\cite{NowheredenseFO:GroheKS17} for the formal definition of first-order logic.
Let $\nu$ be a \emph{vocabulary}, that is, a finite set of relation symbols, each equipped with a prescribed arity.
First-order formulas over $\nu$ are constructed from atomic formulas using Boolean connectives and quantifiers.
The atomic formulas are of the form $x=y$ and $R(x_1,\ldots,x_k)$, where $R$ is a $k$-ary relation symbol in $\nu$.
If $\varphi$ and $\psi$ are first-order formulas over $\nu$, then so are
$\neg\varphi$, $\varphi\land\psi$, $\varphi\lor\psi$, $\varphi\rightarrow\psi$, $\forall x\,\varphi$, and $\exists x\,\varphi$.
For a graph $G$ and a first-order formula $\varphi$, we write $G \models \varphi$ if $G$ satisfies $\varphi$.


To prove \Cref{thm:FPTnowheredense}, it suffices to construct a first-order formula $\phi$ expressing the existence of a parity dominating set of size at most $k$, and then apply the following theorem.

\begin{theorem}[Theorem~1.1 of~\cite{NowheredenseFO:GroheKS17}]\label{FPT:nowheredenseFO}
    Let $\mathcal{C}$ be a nowhere dense graph class.
    For every $\epsilon>0$, there exist a computable function
    $f\colon\mathbb{N}\to\mathbb{N}$ and an algorithm that, given an $n$-vertex graph $G\in\mathcal{C}$ and a first-order formula $\phi$, decides whether $G\models\phi$ in time
    $f(|\phi|)\cdot n^{1+\epsilon}$.
\end{theorem}

\paragraph{Expression of \prbGen.}
Let $(G, f, \Vopen, \Vclose, k)$ denote the instance of \prbGen.
Define
\[
\offV \coloneq \{v\in V(G)\mid f(v)=0\}
\qquad\text{and}\qquad
\onV \coloneq \{v\in V(G)\mid f(v)=1\}.
\]

Consider a vocabulary $\nu = \{E, \Vopen,\Vclose, \offV, \onV\}$, where $E(x,y)$ denotes that $x$ and $y$ are adjacent in $G$.
For each unary predicate symbol $Z\in{\Vopen,\Vclose,\offV,\onV}$, the interpretation of $Z$ is the corresponding subset of $V(G)$; that is, $Z(v)$ holds if and only if $v$ belongs to that subset.

For each $i\in[0,k]$, let $\phi_i(x_1,x_2, \ldots,x_i)$ denote a formula expressing that $D=\{x_1,x_2, \ldots,x_i\}$ is a PD-set of size exactly $i$.
Then the desired formula $\phi$ is given by $\phi \coloneq \bigvee_{i \in [0,k]} \phi_i$.
We show that each $\phi_i$ is a first-order formula over the vocabulary $\nu$.

Throughout the construction, the free variables are assumed to belong to $\{x_1,x_2,\dots, x_i\}$.
We first define a formula expressing that a vertex belongs to $D$:
\begin{align*}
    &\id{selected}(x) \coloneq \bigvee_{j \in [i]} (x=x_j)
\end{align*}

For a positive integer $r$, let $\id{at-least-}r\id{-in-neiop}(x)$ denote the formula expressing that at least $r$ selected vertices belong to $N_G(x)$.
For $r\ge 1$, we define
\begin{align*}
    &\id{at-least-}r\id{-in-neiop}(x)\\
    &\;\coloneq\;
    \exists y_1 \dots \exists y_r \Bigg(
        \bigwedge_{1 \le p < q \le r} \neg(y_p = y_q)
        \;\wedge 
        \bigwedge_{1 \le p \le r}
        \big( \id{selected}(y_p) \wedge E(y_p,x) \big)
    \Bigg).
\end{align*}

For a positive integer $r$, the formula expressing that exactly $r$ selected vertices belong to $N_G(x)$ is
\begin{align*}
    r\id{-in-neiop}(x)\coloneq\id{at-least-}r\id{-in-neiop}(x)\wedge\neg\id{at-least-}(r+1)\id{-in-neiop}(x),
\end{align*}
and 
\begin{align*}
    0\id{-in-neiop}(x)\coloneq\neg\id{at-least-}1\id{-in-neiop}(x).
\end{align*}

Similarly, for the closed neighborhood, we define $\id{at-least-}r\id{-in-neicl}(x)$, $r\id{-in-neicl}(x)$, and $\id{0-in-neicl}(x)$.
For a positive integer $r$,
\begin{align*}
    &\id{at-least-}r\id{-in-neicl}(x)\\
    &\;\coloneq\;
    \exists y_1 \dots \exists y_r \Bigg(
    \bigwedge_{1 \le p < q \le r} \neg(y_p = y_q)
    \;\wedge 
    \bigwedge_{1 \le p \le r}
    \big( \id{selected}(y_p) \wedge (E(y_p,x) \vee y_p = x) \big)
\Bigg), 
\end{align*}
\begin{align*}
    r\id{-in-neicl}(x)
    \;\coloneq\;
    \id{at-least-}r\id{-in-neicl}(x)
    \;\wedge\;
    \neg\,\id{at-least-}(r+1)\id{-in-neicl}(x),
\end{align*}
and 
\begin{align*}
    0\id{-in-neicl}(x)
    \;\coloneq\;
    \neg\,\id{at-least-}1\id{-in-neicl}(x).
\end{align*}

We now define the formula $\id{is-parity-dominated}(v)$ expressing that
$v$ satisfies the parity condition:
\begin{align*}
&\id{is-parity-dominated}(v) 
\;\coloneq\;\\
&\Bigg(
\Vopen(v) \Rightarrow
    \Big(
        (\onV(v) \Rightarrow
            \bigvee_{\substack{1 \le r \le k \\ r \colon\text{odd}}}
            r\id{-in-neiop}(v))
        \wedge  
        (\offV(v) \Rightarrow
            \bigvee_{\substack{0 \le r \le k \\ r \colon \text{even}}}
            r\id{-in-neiop}(v))
    \Big)
\Bigg)
\wedge \\
&\Bigg(
\Vclose(v) \Rightarrow
    \Big(
        (\onV(v) \Rightarrow
            \bigvee_{\substack{1 \le r \le k \\ r\colon \text{odd}}}
            r\id{-in-neicl}(v))
        \wedge 
        (\offV(v) \Rightarrow
            \bigvee_{\substack{0 \le r \le k \\ r\colon \text{even}}}
            r\id{-in-neicl}(v))
    \Big)
\Bigg).
\end{align*}

Finally, the existence of a parity dominating set of size exactly $i$ is expressed by the formula
\begin{align*}
\phi_i \;\coloneq\;
\exists x_1 \dots \exists x_i
\left(
    \bigwedge_{1 \le p < q \le i} \neg(x_p = x_q)
    \;\wedge\;
    \forall v\, (\id{is-parity-dominated}(v))
\right).
\end{align*}

By construction, $\phi_i$ is a first-order formula over the vocabulary $\nu$, and its length is bounded by a function depending only on $k$. 
Therefore, \Cref{thm:FPTnowheredense} follows immediately from \Cref{FPT:nowheredenseFO}.

\section{Omitted discussion for \Cref{thm:NPcomp_planar_bipartite}}
\label{appsec:NPcomp_planar_bipartite}

We now verify that $G$ satisfies the required structural properties.

\begin{claim}
    The graph $G$ is planar, bipartite, has maximum degree~$3$, and girth at least~$g$.
\end{claim}

\begin{proof}
    \rev{
    We first show that $G$ is planar.
    Fix a planar embedding of the incidence graph $G_\phi$.
    By construction, each variable vertex and clause vertex of $G_\phi$ is replaced by a gadget that admits a planar embedding.
    Moreover, every edge of $G_\phi$ is replaced by a single edge joining the corresponding variable gadget and clause gadget.
    Since the connecting vertices of each gadget lie on its outer face, each edge of $G_\phi$ can be replaced without introducing any edge crossing.
    Therefore, $G$ is planar.
    }

    \rev{
    Next, we show that $G$ has maximum degree $3$.
    Every vertex in a variable gadget or a clause gadget has degree at most $3$.
    Moreover, the only vertices incident with edges between distinct gadgets are the vertices $f_i^1$, $f_i^3$, and $\ell_r^j$, each of which has degree exactly $3$ in $G$.
    Hence, every vertex still has degree at most $3$.
    }


    We now consider the girth of $G$.
    Each variable gadget is a path and therefore contains no cycle.
    Each clause gadget contains exactly one cycle, namely the cycle $Y_j$, whose length is $3b$.
    \rev{Since $b \ge g$, every cycle contained entirely in a clause gadget has length at least $g$.
    Consider a cycle $Q$ that is not contained in a single gadget.
    Since variable gadgets are acyclic, $Q$ must contain vertices from some clause gadget $G_{C_j}$ and use at least two of the vertices $\ell_1^j,\ell_2^j,\ell_3^j$.
    Inside $G_{C_j}$, the shortest path between any two of $\ell_1^j,\ell_2^j,\ell_3^j$ has length $b+1$.
    Therefore, $Q$ contains a subpath of length at least $b+1$ within $G_{C_j}$.
    Since $b \ge g$, it follows that $|V(Q)| \ge b+1 > g$.
    Hence every cycle of $G$ has length at least $g$, and thus the girth of $G$ is at least $g$.
    }
    

    Finally, we show that $G$ is bipartite.
    Define a partition $(S,T)$ of $V(G)$ as follows.
    For each variable gadget $G_{x_i}$, place $t_i^1, z_i^1, f_i^2, t_i^3$ in $S$, and place all remaining vertices of $G_{x_i}$ in $T$.
    For each clause gadget $G_{C_j}$, place $c_h$ in $S$ if $h$ is odd and in $T$ if $h$ is even.
    \rev{Since $b=2^a$ is even, the vertices $c_b$, $c_{2b}$, and $c_{3b}$ belong to $T$.}
    Place $\ell_1^j,\ell_2^j,\ell_3^j$ in $S$.
    
    \rev{It is straightforward to verify that every edge inside a variable gadget or a clause gadget joins a vertex of $S$ and a vertex of $T$.
    Furthermore, every edge between distinct gadgets joins a vertex $f_i^1$ or $f_i^3$ in $T$ to a vertex $\ell_r^j$ in $S$.
    Therefore, every edge of $G$ has one endpoint in $S$ and the other in $T$.}
    It follows that $(S,T)$ is a bipartition of $G$.
\end{proof}

To complete our reduction, we show the following lemma.
\NPcplanar*
\label{thm:NPcplanar*}

\begin{proof}
    Suppose first that $\phi$ is a yes-instance of \prbOinT, and  let $\sigma\colon X\rightarrow \{\true,\false\}$ be a satisfying truth assignment such that every clause of $\phi$ contains exactly one true literal.
    
    We construct an ODD-set $D$ of $G$ as follows.
    For each variable $x_i$, if $\sigma(x_i)=\true$, then add $t_i^1,t_i^2,t_i^3$ to $D$; otherwise, add $f_i^1,f_i^2,f_i^3$ to $D$.
    For each clause $C_j$, let $r_j\in \{1,2,3\}$ be the unique index such that the $r_j$-th literal of $C_j$ is true under $\sigma$.
    We then add to $D$ the set $B_{j,r_j}\coloneqq \{c_{r_jb+3p} \mid p=0,1,\ldots,b-1\}$, where indices are taken modulo $3b$ and $c_{3b}$ is identified with $c_0$.
    In other words, in the cycle $Y_j$ we select every third vertex, starting from $c_{r_j\cdot b}$.




    By construction, $|D|=3n+bm=k$.
    It remains to show that $D$ is an ODD-set.

    First consider a variable gadget $G_{x_i}$ for some $i \in [n]$.
    If $\sigma(x_i)=\true$, then $D\cap V(G_{x_i})=\{t_i^1,t_i^2,t_i^3\}$; otherwise, $D\cap V(G_{x_i})=\{f_i^1,f_i^2,f_i^3\}$.
    In either case, every vertex of $G_{x_i}$ is odd-dominated by $D$ (see again \Cref{fig:NPcomp_gadgets}).

    Now consider a clause gadget $G_{C_j}$ for some $j \in [m]$.
    The set $D\cap V(Y_j)$ consists of every third vertex on the cycle $Y_j$.
    Since $\ell^j_1,\ell^j_2,\ell^j_3\notin D$, every vertex of $Y_j\setminus D$ has exactly one neighbor in $D$ on the cycle.
    Hence every vertex of $Y_j$ is odd-dominated by $D$.
    
    It remains to check the three vertices $\ell_1^j,\ell_2^j,\ell_3^j$.
    For the unique true literal vertex $\ell_{r_j}^j$, its neighbor in the variable gadget is not in $D$, while its neighbor $c_{r_jb}$ on the cycle belongs to $D$.
    Thus $\ell_{r_j}^j$ has exactly one neighbor in $D$.
    For each of the other two literal vertices, their neighbor in the variable gadget belongs to $D$, whereas their neighbors on the cycle are not in $D$.
    Hence each of them also has exactly one neighbor in $D$.
    Therefore, every vertex of $G$ is odd-dominated by $D$.

    It remains to verify that the vertices $\ell_1^j,\ell_2^j$, and $\ell_3^j$ are odd-dominated.
    Consider the vertex $\ell_{r_j}^j$ corresponding to the unique true literal of $C_j$.
    Its neighbor in the corresponding variable gadget does not belong to $D$, whereas its neighbor $c_{r_jb}$ on the cycle belongs to $D$. 
    Therefore, $N_G[\ell_{r_j}^j]\cap D=\{c_{r_jb}\}$, and hence $\ell_{r_j}^j$ is odd-dominated.

    For each of the other two vertices among
    ${\ell_1^j,\ell_2^j,\ell_3^j}$, the neighbor in the corresponding variable gadget belongs to $D$, whereas the neighbor on the cycle does not belong to $D$.
    Thus each of these vertices is also odd-dominated.

    Therefore, every vertex of $G$ is odd-dominated by $D$, and hence $D$ is an ODD-set of $G$.

    Conversely, suppose that $(G,k)$ is a yes-instance of \prbOD, and let $D$ be an ODD-set of $G$ with $|D|\leq k$.

    Consider a variable gadget $G_{x_i}$ for some $i\in[n]$. 
    Since $t_i^1$ has degree~$1$ and is adjacent only to $f_i^1$, odd domination of $t_i^1$ requires that exactly one of $t_i^1$ and $f_i^1$ belongs to $D$. 
    Furthermore, among the five vertices $t_i^2,f_i^2,z_i^2,t_i^3$, and $f_i^3$, each vertex of $G_{x_i}$ can odd-dominate at most two of them. 
    Therefore, at least two additional vertices of $G_{x_i}$ must belong to $D$ in order to odd-dominate all five vertices.
    Consequently, we have $|D\cap V(G_{x_i})|\geq 3$.
    

    Next, consider a clause gadget $G_{C_j}$ for some $j\in[m]$.
    Since $Y_j$ is a cycle on $3b$ vertices, every ODD-set of $Y_j$ contains at least $b$ vertices. 
    Hence, we have $|D\cap V(G_{C_j})|\geq b$.

    Summing over all variable and clause gadgets yields $|D| \geq \sum_{i=1}^{n}|D\cap V(G_{x_i})| + \sum_{j=1}^{m}|D\cap V(G_{C_j})| \geq 3n+bm$.
    Since $|D|\leq k=3n+bm$, all inequalities must be tight.
    Therefore, $D$ contains exactly~$3$ vertices from each variable gadget and exactly~$b$ vertices from each clause gadget.
    

    For each $j\in[m]$, the set $D\cap V(G_{C_j})$ contains exactly $b$ vertices. Since all $3b$ vertices of the cycle $Y_j$ are odd-dominated by these $b$ vertices, it follows that $D\cap V(G_{C_j}) \subseteq V(Y_j)$.
    In particular, $\ell_1,\ell_2,\ell_3 \notin D$.

    Consequently, every variable gadget $G_{x_i}$ must odd-dominate all of its eight vertices using only vertices of $G_{x_i}$. 
    Since $|D\cap V(G_{x_i})|=3$, it follows that either $\{t_i^1,t_i^2,t_i^3\}\subseteq D$ or $\{f_i^1,f_i^2,f_i^3\}\subseteq D$.
    We therefore define a truth assignment $\sigma\colon X \to \{\true, \false\}$ by setting 
    \[ 
    \sigma(x_i)= 
    \begin{cases}  
        \true & \text{if } \{t_i^1,t_i^2,t_i^3\}\subseteq D,\\
        \false & \text{otherwise}. 
    \end{cases} 
    \]

    We show that $\sigma$ satisfies every clause of $\phi$ and that each clause contains exactly one true literal.

    


    For each $j\in[m]$, since all $3b$ vertices of $Y_j$ are odd-dominated by the $b$ vertices in $D\cap V(G_{C_j})$, it follows that $D\cap V(Y_j)$ must be one of the following three sets: $ \{c_{b+3p} \mid p\in[b]\}$, $\{c_{2b+3p} \mid p\in[b]\}$, and $\{c_{3b+3p} \mid p\in[b]\}$.
    Without loss of generality, assume that $\{c_{b+3p} \mid p\in[b]\}\subseteq D$.
    Then $c_b\in D$, whereas $c_{2b}, c_{3b}\notin D$.
    Let $x_{q_1}$, $x_{q_2}$, and $x_{q_3}$ be the three variables appearing in $C_j$, and let $\ell^j_1$, $\ell^j_2$, and $\ell^j_3$ denote the corresponding literal vertices of $G_{C_j}$, respectively.
    To odd-dominate $\ell_1^j$, $\ell_2^j$, and $\ell_3^j$, we must have $f_{q_1}^1,f_{q_1}^3\notin D$ and $f_{q_2}^1,f_{q_2}^3, f_{q_3}^1,f_{q_3}^3\in D$.
    Hence, $\sigma(x_{q_1})=\true$ and $\sigma(x_{q_2})=\sigma(x_{q_3})=\false$.
    Therefore, the clause $C_j$ contains exactly one true literal under $\sigma$.


    Since $j$ was arbitrary, every clause of $\phi$ contains exactly one true literal under $\sigma$. 
    Therefore, $\sigma$ is a satisfying truth assignment of $\phi$.
\end{proof}

\section{Omitted discussion for \Cref{thm:hard_k}}\label{appsec:hard_k}

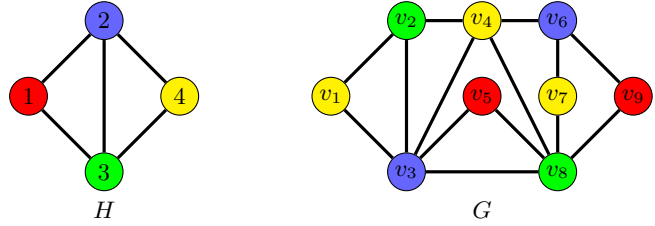
\begin{figure}[ht]
    \centering
    \begin{tikzpicture}

    \node[] at (4,11.5) {$H$};
    \node[draw=black,fill=red, circle, minimum size=5mm, inner sep=0pt] (H1) at (3,13){$1$}; 
    \node[draw=black,fill=blue!60, circle, minimum size=5mm, inner sep=0pt] (H2) at (4,14){$2$};
    \node[draw=black,fill=green, circle, minimum size=5mm, inner sep=0pt] (H3) at (4,12){$3$};
    \node[draw=black,fill=yellow, circle, minimum size=5mm, inner sep=0pt] (H4) at (5,13){$4$};

    \draw[very thick] (H1)--(H2);
    \draw[very thick] (H3)--(H2);
    \draw[very thick] (H1)--(H3);
    \draw[very thick] (H2)--(H4);
    \draw[very thick] (H3)--(H4);

    \node[] at (9,11.5) {$G$};
    \node[draw=black, fill=yellow, circle, minimum size=5mm, inner sep=0pt] (G1) at (7,13){$v_1$};
    \node[draw=black, fill=green, circle, minimum size=5mm, inner sep=0pt] (G2) at (8,14){$v_2$};
    \node[draw=black, fill=blue!60, circle, minimum size=5mm, inner sep=0pt] (G3) at (8,12){$v_3$};
    \node[draw=black, fill=yellow, circle, minimum size=5mm, inner sep=0pt] (G4) at (9,14){$v_4$};
    \node[draw=black, fill=red, circle, minimum size=5mm, inner sep=0pt] (G5) at (9,13){$v_5$};
    \node[draw=black, fill=blue!60, circle, minimum size=5mm, inner sep=0pt] (G6) at (10,14){$v_6$};
    \node[draw=black, fill=yellow, circle, minimum size=5mm, inner sep=0pt] (G7) at (10,13){$v_7$};
    \node[draw=black, fill=green, circle, minimum size=5mm, inner sep=0pt] (G8) at (10,12){$v_8$};
    \node[draw=black, fill=red, circle, minimum size=5mm, inner sep=0pt] (G9) at (11,13){$v_9$};

    \draw[very thick] (G1)--(G2);
    \draw[very thick] (G2)--(G3);
    \draw[very thick] (G1)--(G3);
    \draw[very thick] (G2)--(G4);
    \draw[very thick] (G3)--(G4);
    \draw[very thick] (G3)--(G5);
    \draw[very thick] (G3)--(G8);
    \draw[very thick] (G4)--(G6);
    \draw[very thick] (G4)--(G8);
    \draw[very thick] (G5)--(G8);
    \draw[very thick] (G6)--(G7);
    \draw[very thick] (G6)--(G9);
    \draw[very thick] (G7)--(G8);
    \draw[very thick] (G8)--(G9);    
\end{tikzpicture}
    \caption{An example of an instance $(G,H)$ of \PSI. \rev{Vertices with the same color belong to the same partition set.}}
    \label{fig:instance_PSI}
\end{figure}

\begin{figure}[ht]
    \centering
    \begin{tikzpicture}

    \node[] at (1.5,1.5) {$G_1$};
    \filldraw[fill = red] (1.75,1.65) rectangle (2.25, 1.35);
    
    \node[draw=black, circle, minimum size=5mm, inner sep=0pt] (q) at (0,6){$q_1$};

    \draw[rounded corners=5pt] (0.75,6.3) rectangle (3,5.7);
    \node[] at (1,6) {$V'_1$};
    \node[draw=black, circle, minimum size=5mm, inner sep=0pt] (v5_pr) at (1.5,6){$v'_5$};
    \node[draw=black, circle, minimum size=5mm, inner sep=0pt] (v9_pr) at (2.5,6){$v'_9$};

    \draw[rounded corners=5pt] (0.75,4.8) rectangle (3,4.2);
    \node[] at (1,4.5) {$V_1^*$};
    \node[draw=black, circle, minimum size=5mm, inner sep=0pt] (v5_str) at (1.5,4.5){$v^*_5$};
    \node[draw=black, circle, minimum size=5mm, inner sep=0pt] (v9_str) at (2.5,4.5){$v^*_9$};

    \draw[rounded corners=5pt] (-0.5,3.3) rectangle (4,2.3);
    \node[] at (-0.2,3) {$L_1$};
    \node[draw=black,fill=blue!50, circle, minimum size=3mm, inner sep=0pt, label=below:$\ell_{v_5,1,2}$] (l_1) at (0.5,3){};
    \node[draw=black, fill=green, circle, minimum size=3mm, inner sep=0pt, label=below:$\ell_{v_5,1,3}$] (l_2) at (1.5,3){};
    \node[draw=black, fill=blue!50, circle, minimum size=3mm, inner sep=0pt, label=below:$\ell_{v_9,1,2}$] (l_3) at (2.5,3){};
    \node[draw=black, fill=green, circle, minimum size=3mm, inner sep=0pt, label=below:$\ell_{v_9,1,3}$] (l_4) at (3.5,3){};

    \node[draw=black, circle, minimum size=5mm, inner sep=0pt] (y) at (4,6){$\Tilde{x}_1$};

    \draw[rounded corners=3pt] (4.8,7.5) rectangle (5.2,5.8);
    \node[] at (5,7.3) {$\Tilde{X}_1$};
    \node[draw=black, circle, minimum size=3mm, inner sep=0pt] (y1) at (5,6){};
    \node[] at (5,6.6) {$\vdots$};
    \node[draw=black, circle, minimum size=3mm, inner sep=0pt] (yk) at (5,7){};

    \draw[very thick] (y)--(y1);
    \draw[very thick] (y)--(yk);
    
    \node[draw=black, circle, minimum size=5mm, inner sep=0pt] (x) at (4,4.5){$\Tilde{y}_1$};

    \draw[rounded corners=3pt] (4.8,5) rectangle (5.2,3.3);
    \node[] at (5,4.8) {$\Tilde{Y}_1$};
    \node[draw=black, circle, minimum size=3mm, inner sep=0pt] (x1) at (5,3.5){};
    \node[] at (5,4.1) {$\vdots$};
    \node[draw=black, circle, minimum size=3mm, inner sep=0pt] (xk) at (5,4.5){};

    \draw[very thick] (x)--(x1);
    \draw[very thick] (x)--(xk);

    \draw[very thick] (x)--(y);

    \draw[very thick] (q)--(v5_str);
    \draw[very thick] (q)--(v9_str);

    \draw[very thick] (y)--(v5_str);
    \draw[very thick] (y)--(v9_str);

    \draw[very thick] (x)--(v5_pr);
    \draw[very thick] (x)--(v9_pr);
    \draw[very thick] (x)--(l_1);
    \draw[very thick] (x)--(l_2);
    \draw[very thick] (x)--(l_3);
    \draw[very thick] (x)--(l_4);

    \draw[very thick] (v5_pr)--(v5_str);
    \draw[very thick] (v9_pr)--(v9_str);

    \draw[very thick] (v5_str)--(l_1);
    \draw[very thick] (v5_str)--(l_2);
    \draw[very thick] (v9_str)--(l_3);
    \draw[very thick] (v9_str)--(l_4);

    \node[] at (9.5,1.5) {$G_{1,2}$};
    \node[draw=black, circle, minimum size=5mm, inner sep=0pt] (q12) at (6.5,3){$q_{1,2}$};

    \filldraw[fill = red] (9,1.65-0.35) rectangle (9.5, 1.35-0.35);
    \filldraw[fill = blue!50] (9.5,1.65-0.35) rectangle (10, 1.35-0.35);

    \draw[rounded corners=5pt] (6.9+0.5,5.4+1) rectangle (9.5+0.5,4.6+1);
    \node[] at (7.3+0.5,5+1) {$V_{1,2}^*$};
    \node[draw=black, circle, minimum size=5mm, inner sep=0pt] (z53_str) at (8+0.5,5+1){$z^*_{5,3}$};
    \node[draw=black, circle, minimum size=5mm, inner sep=0pt] (z96_str) at (9+0.5,5+1){$z^*_{9,6}$};

    \draw[rounded corners=5pt] (6.9+0.5,3.9-0.75) rectangle (9.5+0.5,3.1-0.75);
    \node[] at (7.3+0.5,3.5-0.75) {$V'_{1,2}$};
    \node[draw=black, circle, minimum size=5mm, inner sep=0pt] (z53_pr) at (8+0.5,3.5-0.75){$z'_{5,3}$};
    \node[draw=black, circle, minimum size=5mm, inner sep=0pt] (z96_pr) at (9+0.5,3.5-0.75){$z'_{9,6}$};

    \node[draw=black, circle, minimum size=5mm, inner sep=0pt] (y12) at (11,5){$\Tilde{y}_{1,2}$};

    \draw[rounded corners=3pt] (4.7+7,7.6-1) rectangle (5.3+7,5.7-1);
    \node[] at (5+7,7.3-1) {$\Tilde{Y}_{1,2}$};
    \node[draw=black, circle, minimum size=3mm, inner sep=0pt] (y121) at (5+7,6-1){};
    \node[] at (5+7,6.6-1) {$\vdots$};
    \node[draw=black, circle, minimum size=3mm, inner sep=0pt] (y12k) at (5+7,7-1){};

    \draw[very thick] (y12)--(y121);
    \draw[very thick] (y12)--(y12k);
    
    \node[draw=black, circle, minimum size=5mm, inner sep=0pt] (x12) at (4+7,4.5-1){$\Tilde{x}_{1,2}$};

    \draw[rounded corners=3pt] (4.7+7,5.1-1) rectangle (5.3+7,3.1-1);
    \node[] at (5+7,4.8-1) {$\Tilde{X}_{1,2}$};
    \node[draw=black, circle, minimum size=3mm, inner sep=0pt] (x121) at (5+7,3.5-1){};
    \node[] at (5+7,4.1-1) {$\vdots$};
    \node[draw=black, circle, minimum size=3mm, inner sep=0pt] (x12k) at (5+7,4.5-1){};

    \draw[very thick] (x12)--(x121);
    \draw[very thick] (x12)--(x12k);

    \draw[very thick] (x12)--(y12);
    
    \draw[very thick] (x12)--(z53_str);
    \draw[very thick] (x12)--(z96_str);
    \draw[very thick] (y12)--(z53_pr);
    \draw[very thick] (y12)--(z96_pr);
    
    \draw[very thick] (q12)--(z53_str);
    \draw[very thick] (q12)--(z96_str);

    \draw[very thick] (z53_pr)--(z53_str);
    \draw[very thick] (z96_pr)--(z96_str);
    




\end{tikzpicture}
    \caption{Illustrations of a vertex gadget $G_1$ and an edge gadget $G_{1,2}$ constructed from the instance $(G,H)$ in \cref{fig:instance_PSI}.
    For each $ij \in E(H)$, $\Tilde{X}_i$, $\Tilde{Y}_i$, $\Tilde{X}_{i,j}$, and $\Tilde{Y}_{i,j}$ contain $k'+1$ vertices; for readability, $k'-1$ of them are omitted from the figure.
    Each vertex $z^*_{v_iv_j}$ (resp.\ $z'_{v_iv_j}$) is abbreviated as $z^*_{ij}$ (resp.\ $z'_{ij}$).
    \rev{For visual clarity, each vertex $\ell_{v,i,j}$ is colored according to the index $j$.}
    }
    \label{fig:W1_gadgets}
\end{figure}
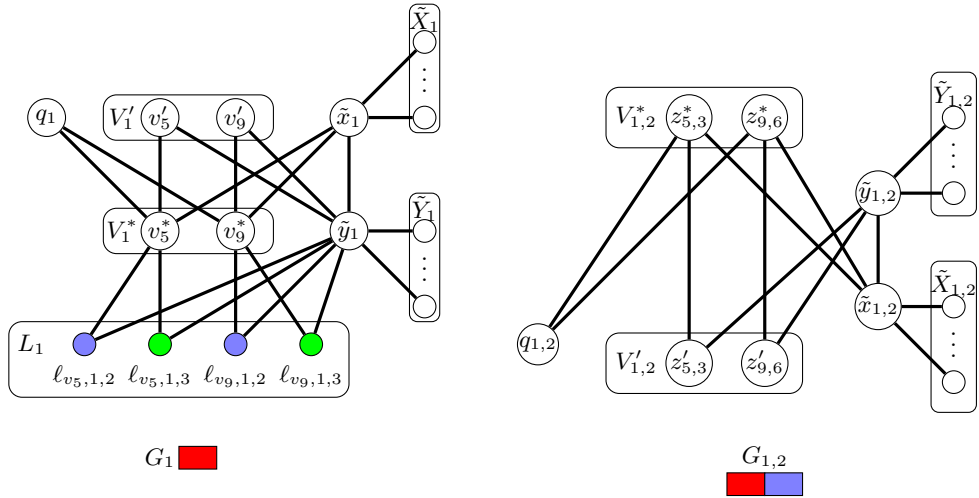

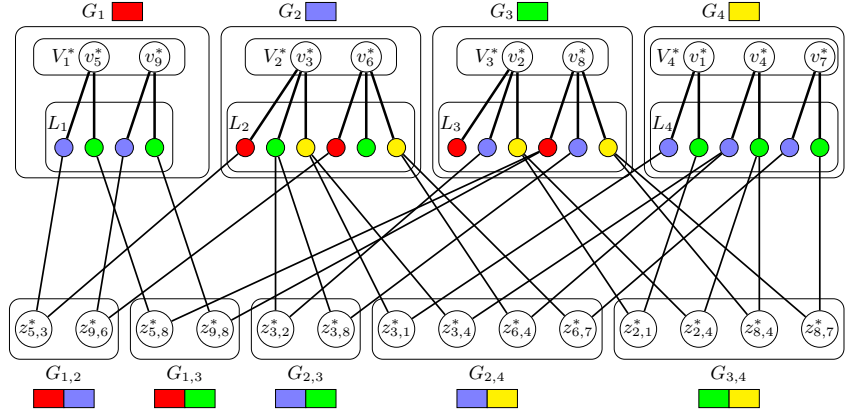
\begin{figure}[ht]
    \centering
    \scalebox{0.8}{\begin{tikzpicture}

    \node[] at (2-0.5, 10.25) {$G_1$};
    \draw[rounded corners=5pt] (0.2,10) rectangle (3.4,7.5);
    \filldraw[fill = red] (1.8, 10.4) rectangle (2.3, 10.1);

    \draw[rounded corners=5pt] (0.5, 9.8) rectangle (3, 9.2);
    \node[] at (1, 9.5) {$V_1^*$};
    \node[draw=black, circle, minimum size=5mm, inner sep=0pt] (v5) at (1.5, 9.5){$v^*_5$};
    \node[draw=black, circle, minimum size=5mm, inner sep=0pt] (v9) at (2.5, 9.5){$v^*_9$};

    \node[] at (0.9, 8.4) {$L_1$};
    \draw[rounded corners=5pt] (0.7,8.75) rectangle (2.8,7.6);
    \node[draw=black, fill=blue!50, circle, minimum size=3mm, inner sep=0pt] (G1_l0) at (1,8){};
    \node[draw=black, fill=green, circle, minimum size=3mm, inner sep=0pt] (G1_l1) at (1.5,8){};
    \node[draw=black, fill=blue!50, circle, minimum size=3mm, inner sep=0pt] (G1_l2) at (2,8){};
    \node[draw=black, fill=green, circle, minimum size=3mm, inner sep=0pt] (G1_l3) at (2.5,8){};

    \draw[very thick] (v5)--(G1_l0);
    \draw[very thick] (v5)--(G1_l1);
    \draw[very thick] (v9)--(G1_l2);
    \draw[very thick] (v9)--(G1_l3);

    \node[] at (5.25-0.5,10.25) {$G_2$};
    \draw[rounded corners=5pt] (3.6,10) rectangle (6.9,7.5);
    \filldraw[fill = blue!50] (5, 10.4) rectangle (5.5, 10.1);

    \draw[rounded corners=5pt] (0+4, 9.8) rectangle (2.5+4, 9.2);
    \node[] at (0.5+4, 9.5) {$V_2^*$};
    \node[draw=black, circle, minimum size=5mm, inner sep=0pt] (v3) at (1+4, 9.5){$v^*_3$};
    \node[draw=black, circle, minimum size=5mm, inner sep=0pt] (v6) at (2+4, 9.5){$v^*_6$};

    \node[] at (-0.1+4,8.4) {$L_2$};
    \draw[rounded corners=5pt] (-0.3+4, 8.75) rectangle (2.8+4, 7.6);
    
    \node[draw=black, fill=red, circle, minimum size=3mm, inner sep=0pt] (G2_l0) at (4+ 0.5*0,8){};
    \node[draw=black, fill=green, circle, minimum size=3mm, inner sep=0pt] (G2_l1) at (4+ 0.5*1,8){};
    \node[draw=black, fill=yellow, circle, minimum size=3mm, inner sep=0pt] (G2_l2) at (4+ 0.5*2,8){};
    \node[draw=black, fill=red, circle, minimum size=3mm, inner sep=0pt] (G2_l3) at (4+ 0.5*3,8){};
    \node[draw=black,fill=green, circle, minimum size=3mm, inner sep=0pt] (G2_l4) at (4+ 0.5*4,8){};
    \node[draw=black, fill=yellow, circle, minimum size=3mm, inner sep=0pt] (G2_l5) at (4+ 0.5*5,8){};

    \draw[very thick] (v3)--(G2_l0);
    \draw[very thick] (v3)--(G2_l1);
    \draw[very thick] (v3)--(G2_l2);
    \draw[very thick] (v6)--(G2_l3);
    \draw[very thick] (v6)--(G2_l4);
    \draw[very thick] (v6)--(G2_l5);

    \node[] at (8.75-0.5,10.25) {$G_3$};
    \draw[rounded corners=5pt] (7.1,10) rectangle (10.4,7.5);
    \filldraw[fill = green] (8.5, 10.4) rectangle (9, 10.1);

    \draw[rounded corners=5pt] (0+7.5, 9.8) rectangle (2.5+7.5, 9.2);
    \node[] at (0.5+7.5, 9.5) {$V_3^*$};
    \node[draw=black, circle, minimum size=5mm, inner sep=0pt] (v2) at (1+7.5, 9.5){$v^*_2$};
    \node[draw=black, circle, minimum size=5mm, inner sep=0pt] (v8) at (2+7.5, 9.5){$v^*_8$};
    
    \node[] at (-0.1+7.5,8.4) {$L_3$};
    \draw[rounded corners=5pt] (-0.3+7.5, 8.75) rectangle (2.8+7.5, 7.6);
    
    \node[draw=black, fill=red, circle, minimum size=3mm, inner sep=0pt] (G3_l0) at (7.5+ 0.5*0,8){};
    \node[draw=black,fill=blue!50, circle, minimum size=3mm, inner sep=0pt] (G3_l1) at (7.5+ 0.5*1,8){};
    \node[draw=black, fill=yellow, circle, minimum size=3mm, inner sep=0pt] (G3_l2) at (7.5+ 0.5*2,8){};
    \node[draw=black,fill=red, circle, minimum size=3mm, inner sep=0pt] (G3_l3) at (7.5+ 0.5*3,8){};
    \node[draw=black, fill=blue!50, circle, minimum size=3mm, inner sep=0pt] (G3_l4) at (7.5+ 0.5*4,8){};
    \node[draw=black, fill=yellow, circle, minimum size=3mm, inner sep=0pt] (G3_l5) at (7.5+ 0.5*5,8){};

    \draw[very thick] (v2)--(G3_l0);
    \draw[very thick] (v2)--(G3_l1);
    \draw[very thick] (v2)--(G3_l2);
    \draw[very thick] (v8)--(G3_l3);
    \draw[very thick] (v8)--(G3_l4);
    \draw[very thick] (v8)--(G3_l5);

    \node[] at (12.25-0.5,10.25) {$G_4$};
    \draw[rounded corners=5pt] (10.6,10) rectangle (13.9,7.5);
    \filldraw[fill = yellow] (12, 10.4) rectangle (12.5, 10.1);

    \draw[rounded corners=5pt] (10.7, 9.8) rectangle (13.8, 9.2);
    \node[] at (11, 9.5) {$V_4^*$};
    \node[draw=black, circle, minimum size=5mm, inner sep=0pt] (v1) at (11.5, 9.5){$v^*_1$};
    \node[draw=black, circle, minimum size=5mm, inner sep=0pt] (v4) at (12.5, 9.5){$v^*_4$};
    \node[draw=black, circle, minimum size=5mm, inner sep=0pt] (v7) at (13.5, 9.5){$v^*_7$};
    
    \node[] at (-0.1+11, 8.4) {$L_4$};
    \draw[rounded corners=5pt] (-0.3+11, 8.75) rectangle (2.8+11, 7.6);

    \node[draw=black, fill=blue!50, circle, minimum size=3mm, inner sep=0pt] (G4_l0) at (11+ 0.5*0,8){};
    \node[draw=black, fill=green, circle, minimum size=3mm, inner sep=0pt] (G4_l1) at (11+ 0.5*1,8){};
    \node[draw=black, fill=blue!50, circle, minimum size=3mm, inner sep=0pt] (G4_l2) at (11+ 0.5*2,8){};
    \node[draw=black, fill=green, circle, minimum size=3mm, inner sep=0pt] (G4_l3) at (11+ 0.5*3,8){};
    \node[draw=black, fill=blue!50, circle, minimum size=3mm, inner sep=0pt] (G4_l4) at (11+ 0.5*4,8){};
    \node[draw=black, fill=green, circle, minimum size=3mm, inner sep=0pt] (G4_l5) at (11+ 0.5*5,8){};

    \draw[very thick] (v1)--(G4_l0);
    \draw[very thick] (v1)--(G4_l1);
    \draw[very thick] (v4)--(G4_l2);
    \draw[very thick] (v4)--(G4_l3);
    \draw[very thick] (v7)--(G4_l4);
    \draw[very thick] (v7)--(G4_l5);    

    \node[] at (1,4.25) {$G_{1,2}$};
    \draw[rounded corners=5pt] (0.1,5.5) rectangle (1.9,4.5);
    \node[draw=black, circle, minimum size=3mm, inner sep=0pt] (G12_1) at (0.5,5){$z^*_{5,3}$};
    \node[draw=black, circle, minimum size=3mm, inner sep=0pt] (G12_2) at (1.5,5){$z^*_{9,6}$};

    \filldraw[fill = red] (0.5, 4) rectangle (1, 3.7);
    \filldraw[fill = blue!50] (1, 4) rectangle (1.5, 3.7);

    \node[] at (1+2,4.25) {$G_{1,3}$};
    \draw[rounded corners=5pt] (0.1+2,5.5) rectangle (1.9+2, 4.5);
    \node[draw=black, circle, minimum size=3mm, inner sep=0pt] (G13_1) at (2.5,5){$z^*_{5,8}$};
    \node[draw=black, circle, minimum size=3mm, inner sep=0pt] (G13_2) at (3.5,5){$z^*_{9,8}$};

    \filldraw[fill = red] (0.5+2, 4) rectangle (1+2, 3.7);
    \filldraw[fill = green] (1+2, 4) rectangle (1.5+2, 3.7);

    \node[] at (1+4,4.25) {$G_{2,3}$};
    \draw[rounded corners=5pt] (0.1+4,5.5) rectangle (1.9+4,4.5);
    \node[draw=black, circle, minimum size=3mm, inner sep=0pt] (G23_1) at (4.5,5){$z^*_{3,2}$};
    \node[draw=black, circle, minimum size=3mm, inner sep=0pt] (G23_2) at (5.5,5){$z^*_{3,8}$};

    \filldraw[fill = blue!50] (0.5+4, 4) rectangle (1+4, 3.7);
    \filldraw[fill = green] (1+4, 4) rectangle (1.5+4, 3.7);

    \node[] at (2+6,4.25) {$G_{2,4}$};
    \draw[rounded corners=5pt] (0.1+6,5.5) rectangle (3.9+6,4.5);
    \node[draw=black, circle, minimum size=3mm, inner sep=0pt] (G24_1) at (6.5,5){$z^*_{3,1}$};
    \node[draw=black, circle, minimum size=3mm, inner sep=0pt] (G24_2) at (7.5,5){$z^*_{3,4}$};
    \node[draw=black, circle, minimum size=3mm, inner sep=0pt] (G24_3) at (8.5,5){$z^*_{6,4}$};
    \node[draw=black, circle, minimum size=3mm, inner sep=0pt] (G24_4) at (9.5,5){$z^*_{6,7}$};

    \filldraw[fill = blue!50] (0.5+7, 4) rectangle (1+7, 3.7);
    \filldraw[fill = yellow] (1+7, 4) rectangle (1.5+7, 3.7);

    \node[] at (2+10,4.25) {$G_{3,4}$};
    \draw[rounded corners=5pt] (0.1+10,5.5) rectangle (3.9+10,4.5);
    \node[draw=black, circle, minimum size=3mm, inner sep=0pt] (G34_1) at (10.5,5){$z^*_{2,1}$};
    \node[draw=black, circle, minimum size=3mm, inner sep=0pt] (G34_2) at (11.5,5){$z^*_{2,4}$};
    \node[draw=black, circle, minimum size=3mm, inner sep=0pt] (G34_3) at (12.5,5){$z^*_{8,4}$};
    \node[draw=black, circle, minimum size=3mm, inner sep=0pt] (G34_4) at (13.5,5){$z^*_{8,7}$};

    \filldraw[fill = green] (0.5+11, 4) rectangle (1+11, 3.7);
    \filldraw[fill = yellow] (1+11, 4) rectangle (1.5+11, 3.7);

    \draw[thick] (G1_l0)--(G12_1)--(G2_l0);
    \draw[thick] (G1_l2)--(G12_2)--(G2_l3);

    \draw[thick] (G1_l1)--(G13_1)--(G3_l3);
    \draw[thick] (G1_l3)--(G13_2)--(G3_l3);

    \draw[thick] (G2_l1)--(G23_1)--(G3_l1);
    \draw[thick] (G2_l1)--(G23_2)--(G3_l4);

    \draw[thick] (G2_l2)--(G24_1)--(G4_l0);
    \draw[thick] (G2_l2)--(G24_2)--(G4_l2);
    \draw[thick] (G2_l5)--(G24_3)--(G4_l2);
    \draw[thick] (G2_l5)--(G24_4)--(G4_l4);

    \draw[thick] (G3_l2)--(G34_1)--(G4_l1);
    \draw[thick] (G3_l2)--(G34_2)--(G4_l3);
    \draw[thick] (G3_l5)--(G34_3)--(G4_l3);
    \draw[thick] (G3_l5)--(G34_4)--(G4_l5);
    
\end{tikzpicture}}
    \caption{Illustration of the graph $G'$ constructed from the instance $(G,H)$ shown in \cref{fig:instance_PSI}. 
    For each vertex gadget $G_i$, all vertices except those in $V_i^* \cup L_i$ are omitted.
    Each vertex $z^*_{v_i v_j}$ (resp.\ $z'_{v_i v_j}$) is abbreviated as $z^*_{ij}$ (resp.\ $z'_{ij}$).
    For each edge gadget $G_{i,j}$, all vertices except those in $V_{i,j}^*$ are omitted.
    }
    \label{fig:bipartite_degeneracy_girth}
\end{figure}


\boundeddg* \label{lem:w_hard_property*}

\begin{proof}
    We first prove that $G'$ is bipartite.
    We define a partition $(A, B)$ of $V(G')$ as follows:
    \[
        A \coloneqq
        \{q_i \mid i\in V(H)\}
        \;\cup\;
        \bigcup_{i\in V(H)} (V_i' \cup L_i \cup \{\tilde{x}_{i}\} \cup \tilde{Y}_i)
        \;\cup\;
        \bigcup_{ij\in E(H)} (V_{i,j}' \cup \{q_{i,j}\} \cup \{\tilde{x}_{i,j}\} \cup \tilde{Y}_{i,j}),
    \]
    and
    \[
        B \coloneqq
        \bigcup_{i\in V(H)} (V_i^{*} \cup \{\tilde{y}_{i}\} \cup \tilde{X}_i)
        \;\cup\;
        \bigcup_{ij\in E(H)} (V_{i,j}^{*} \cup \{\tilde{y}_{i,j}\} \cup \tilde{X}_{i,j}).
    \]
    By construction, every edge of $G'$ has one endpoint in $A$ and the other in $B$. 
    Hence $(A, B)$ is a bipartition of $G'$, and thus $G'$ is bipartite.

    We next prove that $G'$ has degeneracy at most~$4$.
    Recall that a graph has degeneracy at most~$d$ if its vertices admit an ordering in which every vertex has at most~$d$ neighbors appearing later in the ordering.
    Consider the following ordering of $V(G')$:
    \begin{enumerate}
        \item all vertices in $\tilde{X}_i \cup \tilde{Y}_i$ for $i \in V(H)$;
        \item all vertices in $\tilde{X}_{i,j} \cup \tilde{Y}_{i,j}$ for $ij \in E(H)$;
        \item $\bigcup_{ij\in E(H)} V_{i,j}'$;
        \item $\bigcup_{ij\in E(H)} V_{i,j}^{*}$;
        \item all vertices $q_{i,j}$;
        \item all vertices $\{\tilde{x}_{i,j},\tilde{y}_{i,j} \mid ij \in E(H)\}$;
        \item $\bigcup_{i\in V(H)} L_i$;
        \item $\bigcup_{i\in V(H)} V_i'$;
        \item $\bigcup_{i\in V(H)} V_i^{*}$;
        \item all vertices $q_i$ and all vertices $\{\tilde{x}_i,\tilde{y}_i \mid i \in V(H)\}$.
    \end{enumerate}
    Within each item, the ordering is arbitrary.

    We now bound the number of later neighbors of each vertex.
    Each vertex in $V_{i,j}^{*}$ has exactly five neighbors:
    $q_{i,j}$, $\tilde{x}_{i,j}$, one vertex in $V_{i,j}'$, and two vertices of the form $\ell_{u,i,j}$ and $\ell_{v,j,i}$.
    Among these neighbors, only the vertex in $V_{i,j}'$ appears earlier in the ordering.
    Hence every vertex in $V_{i,j}^{*}$ has exactly four later neighbors.

    Every other vertex has at most two neighbors appearing later in the ordering.
    Therefore, every vertex has at most four later neighbors, implying that $G'$ has degeneracy at most~$4$.

    Finally, we show that $G'$ has girth at least~$4$.
    Since $G'$ is bipartite, it contains no odd cycle, and in particular no triangle.
    Moreover, $G'$ is a simple graph, and thus contains no cycle of length~$2$.
    Therefore, $G'$ contains no cycle of length at most~$3$, implying that its girth is at least~$4$.
\end{proof}

We next show that $(G,H)$ is a yes-instance of \PSI\ if and only if $(G',k')$ is a yes-instance of \prbOD.

\medskip
\noindent
\textbf{Sufficiency (``if'' part).}

Suppose that $(G,H)$ is a yes-instance of \PSI, and let $\phi \colon V(H) \to V(G)$ be a colored subgraph isomorphism.
For each $i \in V(H)$, let $u_i \coloneqq \phi(i)$ and define
\[
D' =
\{u_i^{*},u_i',\tilde{x}_i,\tilde{y}_i \mid i \in V(H)\}
\cup
\{z^{*}_{u_i u_j},z'_{u_i u_j},\tilde{x}_{i,j},\tilde{y}_{i,j}
\mid ij \in E(H)\}.
\]
By construction, $|D'| = 4 n_H + 4 m_H = k'$.

We prove that $D'$ is an ODD-set of $G'$.
Fix $i \in V(H)$ and consider the vertex gadget $G_i$.
The vertex $q_i$ is adjacent, among vertices of $D'$, only to $u_i^{*}$, and hence is odd-dominated by $D'$.
Each of $u_i^{*}, u_i', \tilde{x}_i,$ and $\tilde{y}_i$ has exactly three vertices of $D'$ in its closed neighborhood and is thus odd-dominated.

\rev{Now let $j \in N_H(i)$.
Since $\phi$ is a subgraph isomorphism and $ij \in E(H)$, we have $u_i u_j \in E(G)$.
Hence the vertex $z^{*}_{u_i u_j}$ is well-defined.
Moreover, the vertex $\ell_{u_i,i,j}$ is adjacent, among vertices of $D'$, exactly to $u_i^{*}$, $\tilde{y}_i$, and $z^{*}_{u_i u_j}$.
Therefore, $\ell_{u_i,i,j}$ is odd-dominated.}

\rev{It remains to consider the other vertices of $G_i$.
Every vertex in $\tilde{X}_i \cup \tilde{Y}_i$ is a private neighbor\footnote{
A vertex $u$ is called a \emph{private neighbor} of a vertex $v$ if $u$ is adjacent only to $v$.
} of $\tilde{x}_i$ or $\tilde{y}_i$, respectively, and is therefore odd-dominated.
For every vertex $\ell_{v,i,j} \in L_i$ with $v \neq u_i$, the only neighbor of $\ell_{v,i,j}$ in $D'$ is $\tilde{y}_i$.
Similarly, for every vertex $v^{*} \in V_i^{*}$ with $v \neq u_i$, the only neighbors of $v^{*}$ in $D'$ are $\tilde{x}_i$ and vertices of the form $z^{*}_{vw}$, none of which belongs to $D'$ by the definition of $D'$.}
Thus every vertex of $G_i$ is odd-dominated by $D'$.

Next consider an edge gadget $G_{i,j}$ for some $ij \in E(H)$.
\rev{The argument is analogous.
The vertex $q_{i,j}$ is adjacent, among vertices of $D'$, only to $z^{*}_{u_i u_j}$.
Each of $z^{*}_{u_i u_j}$, $z'_{u_i u_j}$, $\tilde{x}_{i,j}$, and $\tilde{y}_{i,j}$ has exactly three vertices of $D'$ in its closed neighborhood.
Every vertex in $\tilde{X}_{i,j} \cup \tilde{Y}_{i,j}$ is a private neighbor of $\tilde{x}_{i,j}$ or $\tilde{y}_{i,j}$, respectively.
Finally, every remaining vertex of $V_{i,j}^{*} \cup V_{i,j}'$ has exactly one neighbor in $D'$.
Therefore, every vertex of $G_{i,j}$ is odd-dominated by $D'$.}

Consequently, every vertex of $G'$ is odd-dominated by $D'$.
Hence $D'$ is an ODD-set of $G'$, and thus $(G',k')$ is a yes-instance of \prbOD.

\medskip
\noindent
\textbf{Necessity (``only if'' part).}

Suppose that $(G',k')$ is a yes-instance of \prbOD, and let $D'$ be an ODD-set of $G'$ with $|D'| \leq k'$.

We first show that, for every vertex gadget $G_i$, both $\tilde{x}_i$ and $\tilde{y}_i$ belong to $D'$.
\rev{Suppose, for contradiction, that $\tilde{x}_i \notin D'$.
Each vertex in $\tilde{X}_i$ is a private neighbor of $\tilde{x}_i$.
Hence every vertex in $\tilde{X}_i$ must belong to $D'$ in order to be odd-dominated.
Since $|\tilde{X}_i| = k'+1$, this contradicts the assumption that $|D'| \leq k'$.
Therefore, $\tilde{x}_i \in D'$.}
By symmetry, $\tilde{y}_i \in D'$ as well.
The same argument applies to every edge gadget $G_{i,j}$, implying that $\tilde{x}_{i,j}, \tilde{y}_{i,j} \in D'$ for all $ij \in E(H)$.

Fix $i \in V(H)$.
We next prove that
\begin{align*}
    |(\{q_i\} \cup V_i^{*} \cup V_i' \cup L_i) \cap D'| \geq 2. 
\end{align*}

\rev{First suppose that $q_i \notin D'$.
Since $\tilde{x}_i \in D'$, the parity condition at $q_i$ implies that at least one vertex $v^{*} \in V_i^{*}$ belongs to $D'$.
Fix such a vertex $v^{*}$.
Consider the corresponding vertex $v' \in V_i'$.
The vertex $v'$ is adjacent only to $v^{*}$ and $\tilde{y}_i$.
Since $\tilde{y}_i \in D'$, the parity condition at $v'$ forces $v' \in D'$.
Hence at least two vertices of $\{q_i\} \cup V_i^{*} \cup V_i' \cup L_i$ belong to $D'$.}


\rev{Next suppose that $q_i \in D'$.
Fix an arbitrary vertex $v^{*} \in V_i^{*}$; such a vertex exists because $V_i \neq \emptyset$.
Since both $q_i$ and $\tilde{x}_i$ belong to $D'$, the parity condition at $v^{*}$ implies that at least one additional vertex in $N(v^{*}) \setminus \{q_i,\tilde{x}_i\} \subseteq V_i' \cup L_i$ belongs to $D'$.
Therefore, $|(\{q_i\} \cup V_i^{*} \cup V_i' \cup L_i) \cap D'| \geq 2$ follows.}
%


Thus every vertex gadget contributes at least two vertices from
$\{q_i\} \cup V_i^{*} \cup V_i' \cup L_i$.
By an analogous argument, every edge gadget $G_{i,j}$ contributes at least two vertices from $\{q_{i,j}\} \cup V_{i,j}^{*} \cup V_{i,j}'$.

Recall that $k' = 4n_H + 4m_H$.
Moreover, we have already shown that $D'$ contains the $2n_H + 2m_H$
forced vertices
\[
\{\tilde{x}_i,\tilde{y}_i \mid i \in V(H)\}
\cup
\{\tilde{x}_{i,j},\tilde{y}_{i,j} \mid ij \in E(H)\}.
\]
\rev{Since there are exactly $n_H$ vertex gadgets and $m_H$ edge gadgets, the above lower bounds imply that each vertex gadget contributes exactly two vertices from $\{q_i\} \cup V_i^{*} \cup V_i' \cup L_i$, and each edge gadget contributes exactly two vertices from $\{q_{i,j}\} \cup V_{i,j}^{*} \cup V_{i,j}'$.}


We now show that $q_i \notin D'$ and $L_i \cap D' = \emptyset$ for every $i \in V(H)$.
\rev{Suppose, for contradiction, that $q_i \in D'$.
As shown above, there exists a vertex in $V_i' \cup L_i$ that also belongs to $D'$.
Moreover, since the parity condition at $q_i$ requires an odd number of neighbors of $q_i$ in $D'$, at least one vertex of $V_i^{*}$ belongs to $D'$.
Consequently,
\[
|(\{q_i\} \cup V_i^{*} \cup V_i' \cup L_i) \cap D'| \geq 3,
\]
contradicting the fact that this quantity is exactly~$2$.
Therefore, $q_i \notin D'$.}

\rev{Since $q_i \notin D'$ and $\tilde{x}_i \in D'$, the parity condition at $q_i$ implies that exactly one vertex of $V_i^{*}$ belongs to $D'$.
Let this vertex be $u_i^{*}$.
Consider the corresponding vertex $u_i' \in V_i'$.
Since $\tilde{y}_i \in D'$, the parity condition at $u_i'$ forces $u_i' \in D'$.
Because the gadget contributes exactly two vertices from $\{q_i\} \cup V_i^{*} \cup V_i' \cup L_i$, it follows that $L_i \cap D' = \emptyset$.}

Hence, for every $i \in V(H)$, there exists a unique vertex $u_i \in V_i$ such that $\{u_i^{*}, u_i'\} \subseteq D'$.
Define $\phi \colon V(H) \to V(G)$ by $\phi(i) \coloneq u_i$.


\rev{Let $ij \in E(H)$ and consider the vertex $\ell_{u_i,i,j}$.
Since $\tilde{y}_i \in D'$ and $\ell_{u_i,i,j} \notin D'$, the parity condition at $\ell_{u_i,i,j}$ implies that $\ell_{u_i,i,j}$ has an odd number of neighbors in $D'$.
Among its neighbors, both $u_i^{*}$ and $\tilde{y}_i$ belong to $D'$.
Hence some additional neighbor of $\ell_{u_i,i,j}$ must belong to $D'$.
By construction, the only remaining neighbors of $\ell_{u_i,i,j}$ are vertices of the form $z^{*}_{u_i x}$ with $x \in V_j$.
Therefore, $z^{*}_{u_i x} \in D'$ for some $x \in V_j$.
Applying the same argument to $\ell_{u_j,j,i}$ yields $z^{*}_{y u_j} \in D'$ for some $y \in V_i$.}

\rev{Since exactly one vertex of $V_{i,j}^{*}$ belongs to $D'$, the two selected vertices above must coincide.
Therefore, $z^{*}_{u_i x} = z^{*}_{y u_j}$, which implies that $x = u_j$ and $y = u_i$.
Consequently, $z^{*}_{u_i u_j} \in D'$.}

\rev{By the definition of the edge gadget, the vertex $z^{*}_{u_i u_j}$ exists only if $u_i u_j \in E(G)$.
Hence $u_i u_j \in E(G)$.
Since this holds for every edge $ij \in E(H)$, the mapping $\phi(i) \coloneqq u_i$ defines a colored subgraph isomorphism from $H$ to $G$.}
\qed


Note that $|V(G')|$ is bounded by a polynomial in $n_H$ and $|V(G)|$.
Recall that $k' = 4n_H + 4m_H = O(m_H)$.

Suppose that \prbOD\ can be solved in time $f(k') \cdot |V(G')|^{o(k'/\log k')}$ for some computable function $f$.
Since $|V(G')|$ is polynomial in $n_H$ and $|V(G)|$, this would yield an algorithm for \PSI\ running in time $f'(m_H) \cdot |V(G)|^{o(m_H/\log m_H)}$ for some computable function $f'$.
\revn{This contradicts the ETH-based lower bound for \PSI.}
Therefore, no such algorithm for \prbOD\ exists, which completes the proof of \Cref{thm:hard_k}.

\section{Omitted discussion for \Cref{thm:hard_split}}\label{app:split_hard}

\thmhardsplit* \label{thm:hard_split*}

We present a polynomial-time reduction from \prbOD\ on general graphs to \prbOD\ on split graphs.

Let $(G,k)$ be an instance of \prbOD.
We may assume without loss of generality that $k$ is odd.
Indeed, if $k$ is even, then we obtain an equivalent instance $(G^{*},k^{*})$ by adding an isolated vertex to $G$ and setting $k^{*} \coloneqq k+1$.
Since every isolated vertex must belong to every ODD-set, $(G,k)$ is a yes-instance if and only if $(G^{*},k^{*})$ is.

Let $V(G)=\{v_1,v_2,\dots,v_n\}$.
For each $i\in[k+1]$, let $V^i \coloneqq  \{v_1^i,v_2^i,\ldots,v_n^i\}$.
\rev{We construct a graph $G'$ as follows.}

The vertex set of $G'$ consists of $V(G') \coloneqq V(G) \cup \{z\} \cup \bigcup_{i \in [k+1]} V^{i}$, where $z$ is a newly introduced vertex.
The graph $G'$ contains all edges with both endpoints in $V(G) \cup \{z\}$, and hence this set induces a clique.
Moreover, the set $I \coloneqq \bigcup_{i \in [k+1]} V^{i}$
is an independent set.
Finally, for each $i \in [k+1]$, each vertex $v_b^{i} \in V^{i}$, and each vertex $v_a \in N_G[v_b]$, add the edge $v_a v_b^{i}$.


Set $k' \coloneq k$.
This completes the construction of the instance $(G',k')$.

By construction, $G'$ is a split graph since $V(G)\cup\{z\}$ is a clique and $I$ is an independent set.
Moreover, the reduction clearly can be done in polynomial time.

Next, we prove the correctness of the reduction.
The following lemma completes the proof of \cref{thm:hard_split}.

\begin{lemma}
    The instance $(G,k)$ is a yes-instance of \prbOD\ if and only if the instance $(G',k')$ is a yes-instance of \prbOD.
\end{lemma}

\begin{proof}
    Suppose that $(G,k)$ is a yes-instance of \prbOD.
    Then there exists an ODD-set $D$ of $G$ with $|D| \leq k$.
    We distinguish two cases according to the parity of $|D|$.

    First, suppose that $|D|$ is odd.
    Define $D' \coloneq D$.
    Since $D \subseteq V(G)$ and $V(G)\cup\{z\}$ induces a clique in $G'$, every vertex in $C \coloneqq V(G)\cup\{z\}$ is odd-dominated by $D'$.
    Now consider a vertex $v_j^i \in I$.
    By construction, $N_{G'}(v_j^i)=N_G[v_j]$.
    Since $D$ is an ODD-set of $G$, the vertex $v_j$ is odd-dominated by $D$ in $G$, and hence $v_j^i$ is odd-dominated by $D'$ in $G'$.
    Therefore, every vertex of $G'$ is odd-dominated by $D'$, implying that $D'$ is an ODD-set of $G'$.
    Because $|D'|=|D|\leq k=k'$, the instance $(G',k')$ is a yes-instance of \prbOD.

    Next, suppose that $|D|$ is even.
    Define $D' \coloneqq D \cup \{z\}$.
    Since $k$ is odd and $|D|$ is even, we have $|D| \leq k-1$, and therefore $|D'| \leq k = k'$.
    Moreover, $|D'|$ is odd.

    As above, since $D' \subseteq C$ and $C$ is a clique, every vertex in $C$ is odd-dominated by $D'$.
    Now consider a vertex $v_j^i \in I$.
    Because $z$ is adjacent to no vertex in $I$, we have $N_{G'}(v_j^i)\cap D' = N_G[v_j]\cap D$.
    Since $v_j$ is odd-dominated by $D$ in $G$, it follows that $v_j^i$ is odd-dominated by $D'$ in $G'$.
    Hence every vertex of $G'$ is odd-dominated by $D'$.
    Therefore, $D'$ is an ODD-set of $G'$, and $(G',k')$ is a yes-instance of \prbOD.
    
    Suppose that $(G',k')$ is a yes-instance of \prbOD.
    Then there exists an ODD-set $D'$ of $G'$ with $|D'| \leq k'$.
    
    We first show that $D' \cap I = \emptyset$.
    Suppose, for contradiction, that there exists a vertex $v_j^i \in D' \cap I$.
    Since $|D'|\leq k = k'$ and there are $k+1$ copies of each vertex $v_j$, there exists an index $r \in [k+1]\setminus \{i\}$ such that $v_j^r \notin D'$.
    By construction, $N_{G'}(v_j^i)=N_{G'}(v_j^r)$.
    Hence
    \[
    N_{G'}[v_j^i]\cap D'
    =
    \bigl(N_{G'}(v_j^r)\cap D'\bigr)\cup\{v_j^i\}.
    \]
    Therefore, exactly one of the vertices $v_j^i$ and $v_j^r$ has an odd number of vertices of $D'$ in its closed neighborhood, contradicting the assumption that $D'$ is an ODD-set.
    Thus, $D' \cap I = \emptyset$.

    Define $D \coloneqq D'\setminus\{z\}$.
    Since $D'\cap I=\emptyset$, we have $D \subseteq V(G)$.
    We claim that $D$ is an ODD-set of $G$.
    Fix an arbitrary vertex $v_j \in V(G)$.
    Choose any copy $v_j^i \in I$.
    Since $v_j^i$ is odd-dominated by $D'$ in $G'$ and $N_{G'}(v_j^i)=N_G[v_j]$, we obtain
    \[
    |N_G[v_j]\cap D|
    =
    |N_{G'}(v_j^i)\cap D'|
    \equiv 1 \pmod 2.
    \]
    Therefore, $v_j$ is odd-dominated by $D$ in $G$.
    Since this holds for every vertex $v_j \in V(G)$, the set $D$ is an ODD-set of $G$.
    
    Finally, since $|D|\leq|D'|\leq k'=k$, the instance $(G,k)$ is a yes-instance of \prbOD.
\end{proof}

\section{FPT algorithm parameterized by $k$ for graphs of girth at least~$5$}\label{appsec:FPTgirth5}
In this section, we prove that \prbOD\ is fixed-parameter tractable when restricted to graphs of girth at least~$5$.

\begin{lemma}\label{lem:manyONsafe}
    \Cref{rule:manyON} is safe.
\end{lemma}

\begin{proof}
    Let $(G,F,f,k)$ be an instance to which \Cref{rule:manyON} applies, and let $v$ be the vertex identified by the rule.
    Let $(G,F',f',k)$ denote the resulting instance, where $F' = F \cup \{v\}$.

    Suppose $(G,F,f,k)$ is a yes-instance, and let $D \subseteq V(G)\setminus F$ be a PD-set for $(G,F,f,k)$ with $|D \cup F|\leq k$.
    We claim that $v \in D$.
    Assume for contradiction that $v \notin D$.
    For every vertex $x \in N_G(v)\cap f^{-1}(1)$, the parity constraint at $x$ implies $|N_G(x)\cap D| \geq 1$.
    Choose an arbitrary vertex $u_x \in N_G(x)\cap D$.

    Since the girth of $G$ is at least~$5$, no two distinct vertices of $N_G(v)$ have a common neighbor other than $v$.
    Hence, for distinct $x,y \in N_G(v)\cap f^{-1}(1)$, we have $u_x \neq u_y$.
    Therefore $|D| \geq |N_G(v) \cap f^{-1}(1)| > k$, contradicting $|D| \leq |D \cup F|\le k$.
    Thus $v \in D$.

    Define $D' \coloneq D \setminus \{v\}$. 
    Then $D' \subseteq V(G)\setminus F'$ and $|D' \cup F'| = |D \cup F| \le k$.
    By the definition of $f'$, the set $D'$ satisfies all parity constraints in $(G,F',f',k)$.
    Hence $(G,F',f',k)$ is a yes-instance.
    
    Conversely, suppose $(G,F',f',k)$ is a yes-instance, and let $D' \subseteq V(G)\setminus F'$ with $|D' \cup F'|\le k$ be a PD-set for $(G,F',f',k)$.
    Since $v \in F'$, we have $v \notin D'$.
    Define $D \coloneq D' \cup \{v\}$.
    Then $D \subseteq V(G)\setminus F$ and $|D \cup F| = |D' \cup F'| \le k$.

    Let $x \in V(G)$.
    If $x \notin N_G[v]$, then $f'(x)=f(x)$ and $N_G(x)\cap D = N_G(x)\cap D'$, and the parity constraint at $x$ follows immediately.
    If $x \in N_G[v]$, then $f'(x)=\overline{f(x)}$ and $|N_G(x)\cap D| = |N_G(x)\cap D'| + 1$.
    Since $|N_G(x)\cap D'| + f'(x) \equiv 0 \pmod 2$, it follows that $|N_G(x)\cap D| + f(x) \equiv 0 \pmod 2$.
    Thus $D$ satisfies all parity constraints, and $(G,F,f,k)$ is a yes-instance.
\end{proof}

The following lemma is central to our algorithm.

\girthdegree*\label{lem:girth_degree*}

\begin{proof}
    Let $v \in V(G) \setminus F$. 
    Since \Cref{rule:manyON} is no longer applicable, we have $|N_G(v) \cap f^{-1}(1)| \leq k$.
    
    We bound $|N_G(v)\cap f^{-1}(0)|$.
    \revn{Observe that every vertex $u$ with $f(u)=0$ is contained in $N_G[w]$ for some vertex $w\in F$.}
    Hence every such $u$ is adjacent to at least one vertex of $F$.
    
    Fix $v \in V(G) \setminus F$.
    Since $G$ has girth at least~$5$, no two distinct neighbors of $v$ have a common neighbor other than $v$.
    Therefore, each vertex $w \in F$ is adjacent to at most one vertex in $N_G(v)$.
    Consequently, $|N_G(v) \cap f^{-1}(0)| \le |F| \leq k$, where the last inequality follows from \Cref{rule:kzero}.

    Combining the bounds, $\deg_G(v) = |N_G(v) \cap f^{-1}(1)| + |N_G(v) \cap f^{-1}(0)| \le k + k = 2k$.
\end{proof}


\subsection{Derandomization of the Auxiliary Problem}

To derandomize our algorithm, we employ \emph{universal sets}; see, for example, Cygan et al.~\cite{book:CyganFKLMPPS15}.
Let $V$ be an $n$-element set.
An \emph{$(n,k)$-universal set} over $V$ is a family $\mathcal{U} \subseteq 2^V$ such that for every $S \subseteq V$ with $|S|=k$, the collection $\{A \cap S \mid A \in \mathcal{U}\}$ equals $2^S$. 
Naor et al.~\cite{derandmize:NaorSS95} provided a deterministic construction of an $(n,k)$-universal set $\mathcal{U}$ of size $|\mathcal{U}| = 2^k \cdot k^{O(\log k)} \cdot \log n$ in $O(n \cdot |\mathcal{U}|)$ time.



\begin{lemma}\label{lem:deterministic}
    There exists a deterministic algorithm that solves the auxiliary problem in $2^{O(k^3)} \cdot n^{O(1)}$ time.
\end{lemma}

\begin{proof}
    To derandomize the algorithm, we set $q \coloneq 4k^3 + 2k^2$ and construct an $(|E(G)|, q)$-universal set $\mathcal{U} \subseteq 2^{E(G)}$. 
    For each $U \in \mathcal{U}$, we define an edge coloring $\chi_U \colon E(G) \to \{\redcol, \bluecol\}$ such that $\chi_U(e) = \redcol$ if $e \in U$ and $\chi_U(e) = \bluecol$ otherwise. 
    The properties of $\mathcal{U}$ ensure that for any solution $S$ with $|\Lambda_S \cup \Gamma_S| \le q$, there exists at least one $U^* \in \mathcal{U}$ such that $\chi_{U^*}$ colors all edges in $\Lambda_S$ red and all edges in $\Gamma_S$ blue.

    By the construction of Naor et al.~\cite{derandmize:NaorSS95}, the size of the universal set is $|\mathcal{U}| = \rev{2^q \cdot q^{O(\log q)}} \cdot \log |E(G)| = 2^{O(k^3)} \cdot \mathrm{poly}(k, \log n)$.
    Since each procedure of the dynamic programming algorithm takes $2^{O(k^2)} \cdot n^{O(1)}$ time, the total deterministic running time is $|\mathcal{U}| \cdot 2^{O(k^2)} \cdot n^{O(1)} = 2^{O(k^3)} \cdot n^{O(1)}$.
\end{proof}

Combined with \Cref{lem:deterministic}, since the remaining steps of the algorithm are deterministic, the overall running time is $2^{O(k^3)} \cdot n^{O(1)}$. 
This completes the proof of \Cref{thm:FPT_girth5}.
\section{Omitted discussion for \cref{subsec:clique-width}}\label{appsec:clique-width}
The \emph{clique-width} of a graph $G$, denoted by $\cw(G)$, is the minimum integer $w$ such that $G$ can be constructed using at most $w$ labels by repeatedly applying the following operations:
\begin{itemize}
    \item creating a new graph consisting of a single vertex $v$ with an arbitrary label $i$;
    \item taking the disjoint union of two labeled graphs;
    \item relabeling all vertices with label $i$ to label $j$;
    \item adding all possible edges between vertices of two distinct labels $i$ and $j$.
\end{itemize}

Such a construction naturally induces a rooted tree in which each node corresponds to one of the above operations.
This tree is called a \emph{$w$-expression tree} of $G$.
Nodes corresponding to the four operations are referred to as an \emph{introduce node} $\circ_i$, a \emph{union node} $\cup$, a \emph{relabel node} $\rho_{i \to j}$, and a \emph{join node} $\eta_{i,j}$, respectively.
In any $w$-expression tree, every introduce node is a leaf, and conversely, every leaf is an introduce node.

It is known that, given a graph $G$, one can compute a $(2^{\cw(G)+1}-1)$-expression for $G$ in $O(|V(G)|^3)$ time~\cite{cw:HlinenyO08,cw:Oum08,cw:OumS06}.
A $w$-expression tree is called \emph{irredundant} if every edge of the represented graph is introduced by exactly one join node.
Any $w$-expression tree for an $n$-vertex graph can be transformed into an irredundant one with $O(n)$ nodes in linear time~\cite{cw:CourcelleO00}.
Thus, throughout this paper, we assume that all $w$-expression trees are irredundant.

For a node $t$ in a $w$-expression tree $T_G$, let $G_t$ denote the labeled graph represented by the subtree rooted at $t$.
In a labeled graph, a vertex with label $i$ is called an \emph{$i$-vertex}.
For each $i \in [w]$, let $U_i^t$ denote the set of $i$-vertices in $G_t$.

We now state the main result of this subsection.

\fptcliquew* \label{thm:cliquew*}

Let $(G,f,\Vopen,\Vclose,k)$ be an instance of \prbGen, where $G$ has $n$ vertices and $T_G$ is a $w$-expression tree of $G$.
We solve the problem by performing a bottom-up dynamic programming over $T_G$, where the update rules depend on the type of each node.

For a node $t$ of $T_G$, a pair of functions $(\alpha,\beta)$ with $\alpha,\beta \colon [w] \to \{0,1\}$ is called a \emph{characteristic} of a set $D \subseteq V(G_t)$ if:
\begin{itemize}
    \item for every $i \in [w]$, $\alpha(i) \equiv_2 |D \cap U^t_i|$;
    \item for every $i \in [w]$, the following \emph{parity conditions} hold: for all $v \in U_i^t \cap \Vopen$, we have $|N_{G_t}(v) \cap D| + f(v) \equiv_2 \beta(i) $, and for all $v \in U_i^t \cap \Vclose$, we have $|N_{G_t}[v] \cap D| + f(v) \equiv_2 \beta(i)$.
\end{itemize}

For each node $t$ of $T_G$, we maintain a DP table indexed by states $s = (\alpha,\beta)$.
For each state $s$, the table entry $c_t[s]$ stores the minimum size of a PD-set $D \subseteq V(G_t)$ whose characteristic is $(\alpha,\beta)$.
If no such set exists, we set $c_t[s] = +\infty$.
Since there are at most $2^{w}$ choices for each of $\alpha$ and $\beta$, the number of states for each node is at most $4^{w}$.

After all table entries $c_t[s]$ have been computed, we check whether there exists a function $\alpha$ such that $c_r[(\alpha,\beta^*)] \le k$, where $r$ is the root of $T_G$ and $\beta^*(i) = 0$ for all $i \in [w]$.

We now give an intuitive explanation of the states of our DP table.
The value $\alpha(i)$ describes the parity of the number of vertices of label $i \in [w]$ selected into a solution.
The value $\beta(i)$ encodes the \emph{parity-domination status} of label $i$, namely, whether all $i$-vertices satisfy the parity condition or all violate it.
If $U_i^t = \emptyset$, we allow $\beta(i)$ to take either value in $\{0,1\}$.

\revn{
The correctness of the above definition follows from the following observation.
For every node $t$ and every label $i \in [w]$, all vertices in $U_i^t$ satisfy the same parity condition.
Indeed, any two vertices with the same label are treated identically by all operations that may appear above $t$ in the expression tree.
Hence, for any two vertices $u,v \in U_i^t$ and any set $D \subseteq V(G_t)$, the values $|N_{G_t}(u)\cap D|+f(u)$ and $|N_{G_t}(v)\cap D|+f(v)$ are affected in exactly the same way by every subsequent join and relabel operation.
Therefore, these two values are equal modulo~$2$ at every stage of the computation.
Consequently, all vertices of label $i$ either satisfy the parity condition or violate it simultaneously.
It is therefore sufficient to represent this information by a single bit $\beta(i)\in\{0,1\}$.
}


We now explain the update rule for each node.

\paragraph{Introduce node $\circ_i$.}
Let $t = \circ_i$ be an introduce node that introduces a single vertex $v$ with label $i$.
We may choose whether to include $v$ in the solution.
Since $G_t$ contains no vertices with label $j \neq i$, if there exists some $j \in [w] \setminus \{i\}$ with $\alpha(j)=1$, then we set $c_t[\alpha,\beta]=+\infty$.
For $j \in [w] \setminus \{i\}$, the value of $\beta(j)$ is irrelevant, as the corresponding parity condition is always satisfied.

If $v \in \Vopen$, we define
\begin{align*}
	c_t[\alpha,\beta]=
	\begin{cases}
		0
		 & (\text{if } \alpha(i)=0,\ \beta(i)\equiv_2 f(v),\ \alpha(j)=0\ (\forall j \in [w] \setminus \{i\})) \\
		1
		 & (\text{if } \alpha(i)=1,\ \beta(i)\equiv_2 f(v),\ \alpha(j)=0\ (\forall j \in [w] \setminus \{i\})) \\
		+\infty
		 & \text{(otherwise.)}
	\end{cases}
\end{align*}

If $v \in \Vclose$, then including $v$ in the solution flips its parity condition.
Accordingly, we define
\begin{align*}
	c_t[\alpha,\beta]=
	\begin{cases}
		0
		 & (\text{if } \alpha(i)=0,\ \beta(i)\equiv_2 f(v),\ \alpha(j)=0\ \ (\forall j \in [w] \setminus \{i\}))   \\
		1
		 & (\text{if } \alpha(i)=1,\ \beta(i)\equiv_2 \overline{f(v)},\ \alpha(j)=0\ \ (\forall j \in [w] \setminus \{i\})) \\
		+\infty
		 & \text{(otherwise.)}
	\end{cases}
\end{align*}

\paragraph{Union node $\cup$.}
Let $t$ be a union node with children $t_1$ and $t_2$.
Suppose that $D_1 \subseteq V(G_{t_1})$ and $D_2 \subseteq V(G_{t_2})$ have characteristics $(\alpha_1,\beta_1)$ and $(\alpha_2,\beta_2)$, respectively.
\revn{
For each label $i\in[w]$, vertices of label $i$ occurring in $G_{t_1}$ and $G_{t_2}$ must satisfy the same parity condition.
Hence, we only combine states satisfying $\beta_1(i)=\beta_2(i)$ for all $i\in[w]$.
}
Then the characteristic of $D_1 \cup D_2$ is $(\alpha,\beta)$, where $\alpha(i)\equiv_2 \alpha_1(i)+\alpha_2(i)$ and $\beta(i)=\beta_1(i)=\beta_2(i)$ for all $i \in [w]$.
Thus, we define
\begin{align*}
	c_t[\alpha,\beta]
	=
	\min_{\substack{
			\alpha_1,\alpha_2,\beta_1,\beta_2                            \\
			\alpha(i)\equiv_2 \alpha_1(i)+\alpha_2(i)\ (\forall i \in [w]) \\
			\beta(i)=\beta_1(i)=\beta_2(i)\ (\forall i \in [w])
		}}
	\left(
	c_{t_1}[\alpha_1,\beta_1]
	+
	c_{t_2}[\alpha_2,\beta_2]
	\right).
\end{align*}

\paragraph{Relabel node $\rho_{i \to j}$.}
Let $t = \rho_{i \to j}$ be a relabel node with its child $t'$.
After relabeling, no vertex with label $i$ remains, and therefore we must set $c_t[\alpha,\beta]=+\infty$ whenever $\alpha(i)=1$.
Moreover, $i$-vertices and $j$-vertices in $G_{t'}$ are merged into label $j$ in $G_t$.
Thus, a subset $D' \subseteq V(G_{t'})$ with characteristic $(\alpha',\beta')$ induces a subset of $G_t$ with characteristic $(\alpha,\beta)$ satisfying $\alpha(j)\equiv_2 \alpha'(i)+\alpha'(j)$ and $\beta(j)=\beta'(i)=\beta'(j)$.

All other labels remain unchanged.
Thus, we set
\begin{align*}
	c_t[\alpha,\beta]=
	\begin{cases}
		+\infty
		 & (\text{if } \alpha(i)=1)                               \\[1ex]
		\displaystyle
		\min_{\substack{
				\alpha',\beta'                                            \\
				\alpha(k)=\alpha'(k)\ (\forall k \in [w] \setminus \{i,j\}) \\
				\beta(k)=\beta'(k)\ (\forall k \in [w] \setminus \{j\})   \\
				\alpha(j)\equiv_2 \alpha'(i)+\alpha'(j)                     \\
				\beta(j)=\beta'(i)=\beta'(j)
			}
		}
		c_{t'}[\alpha',\beta']
		 & \text{(otherwise.)}
	\end{cases}
\end{align*}

\paragraph{Join node $\eta_{i,j}$.}
Let $t=\eta_{i,j}$ be a join node with its child $t'$.
The join operation adds all edges between $i$-vertices and $j$-vertices in $G_{t'}$.
Let $D' \subseteq V(G_{t'})$ have characteristic $(\alpha',\beta')$.
Then each $i$-vertex gains $|D' \cap U_j^{t'}|\equiv_2 \alpha'(j)$ new neighbors in $D'$, and each $j$-vertex gains $|D' \cap U_i^{t'}|\equiv_2 \alpha'(i)$ new neighbors in $D'$.
\revn{
Since the expression tree $T_G$ is irredundant, there are no edges between labels $i$ and $j$ before applying $\eta_{i,j}$.
Hence every $i$-vertex receives exactly $|D'\cap U_j^{t'}|$ new neighbors in $D'$, and symmetrically every $j$-vertex receives exactly $|D'\cap U_i^{t'}|$ new neighbors in $D'$.
}
Thus, the parity conditions are updated, and we define
\begin{align*}
	c_t[\alpha,\beta]
	=
	\min_{\substack{
			\alpha',\beta'                                            \\
			\alpha(k)=\alpha'(k)\ (\forall k \in [w]) \\
			\beta(k)=\beta'(k)\ (\forall k \in [w]\setminus \{i,j\})    \\
			\beta(i)\equiv_2 \beta'(i)+\alpha'(j)                       \\
			\beta(j)\equiv_2 \beta'(j)+\alpha'(i)
		}}
	c_{t'}[\alpha',\beta'].
\end{align*}

\paragraph{Complexity Analysis.}
For each node of the $w$-expression tree $T_G$, the DP table contains at most $4^w$ states.

Each table entry for introduce, relabel, and join nodes can be computed in time polynomial in $w$. 
To evaluate the transition at a union node efficiently, we use a standard reduction from min-plus XOR convolution to ordinary XOR convolution.

Let $N$ be a positive integer, and let $f,g \colon \mathbb{F}_2^N \to \{0,1,\ldots,M\}$. 
Our goal is to compute the min-plus XOR convolution
\[
    H(c) \coloneq \sum_{a + b = c} F(a) \cdot G(b),
\]
where addition is taken component-wise over $\mathbb{F}_2$.
Introduce an indeterminate $y$ and define $F(a)=y^{f(a)}$ and $G(b)=y^{g(b)}$.
Consider their ordinary XOR convolution
\[
H(c) = \sum_{a+b=c} F(a)G(b) = \sum_{a+b=c} y^{f(a)+g(b)}.
\]
The minimum degree of the polynomial $H(c)$ is exactly $\min_{a+b=c} \bigl( f(a)+g(b) \bigr)$.
Hence, computing the min-plus XOR convolution reduces to computing an ordinary XOR convolution over the polynomial ring $\mathbb{Z}[y]$. 
Since every polynomial has degree at most $n+1$, additions and multiplications over $\mathbb{Z}[y]$ require $O(M)$ and $O(M^2)$ time, respectively. 
Using the Fast Walsh--Hadamard Transform, all values of the ordinary XOR convolution can be computed using $O(N2^N)$ ring operations. 
Consequently, all values of the min-plus XOR convolution can be computed simultaneously in $O(M^2 N2^N)$ time.

We apply this procedure with $N=w$. 
For every fixed function $\beta$, define $f(\alpha)=c_{t_1}[\alpha,\beta]$ and $g(\alpha)=c_{t_2}[\alpha,\beta]$, where every occurrence of $+\infty$ is replaced by $n+1$, since every feasible solution has size at most $n$.
In other words, set $M = n + 1$.
The above algorithm computes all values $c_t[\alpha,\beta]$ simultaneously in $O(n^2w2^w)$ time.
Since there are $2^w$ possible choices of $\beta$, all table entries for a union node can be computed in $O(n^2w4^w)$ time.

Since the $w$-expression tree has $O(n)$ nodes, the overall running time is $O(n^3w4^w)$.
In particular, given a $w$-expression tree of $G$, the problem \prbGen\ can be solved in $4^w\cdot n^{O(1)}$ time.
\section{Definitions of Graph Parameters}

The \emph{vertex cover number} of a graph $G$, denoted by $\vc(G)$, is the minimum size of a vertex set $S \subseteq V(G)$ such that every edge of $G$ has at least one endpoint in $S$.

The \emph{tree-depth} of a graph $G$, denoted by $\td(G)$, is defined recursively as follows.
If $|V(G)| = 1$, then $\td(G) = 1$.
If $G$ is disconnected, then $\td(G)$ is the maximum tree-depth among its connected components.
Otherwise,
\[
\td(G) = 1 + \min_{v \in V(G)} \td(G - v),
\]
where $G - v$ denotes the graph obtained from $G$ by deleting the vertex $v$.
Equivalently, the tree-depth of $G$ is the minimum depth of an \emph{elimination forest} of $G$.
An elimination forest of $G$ is a rooted forest $F$ on $V(G)$ such that every edge of $G$ joins a vertex with one of its ancestors in $F$.
The depth of $F$ is the maximum number of vertices on a path from a root to a vertex.

The \emph{feedback vertex set number} of a graph $G$, denoted by $\fvs(G)$, is the minimum size of a vertex set $S \subseteq V(G)$ such that $G - S$ is acyclic.

The \emph{cluster deletion number} of a graph $G$, denoted by $\cvd(G)$, is the minimum size of a vertex set $S \subseteq V(G)$ such that $G - S$ is a disjoint union of cliques.

The following simple observation will be used repeatedly.
For each of the parameters $\vc$, $\td$, $\fvs$, and $\cvd$, adding a single vertex to a graph increases the parameter value by at most~$1$.
\section{Omitted discussion for \Cref{subsec:kernel}}

Let $Q \subseteq \Sigma^{*} \times \mathbb{N}$ be a parameterized problem.
A \emph{kernelization} for $Q$ is a polynomial-time algorithm $\mathcal{A} \colon \Sigma^{*} \times \mathbb{N} \to \Sigma^{*} \times \mathbb{N}$ such that for every $(x,k) \in \Sigma^{*} \times \mathbb{N}$:
\begin{enumerate}
    \item $(x,k) \in Q$ if and only if $\mathcal{A}(x,k) \in Q$, and
    \item $|\mathcal{A}(x,k)| \le f(k)$, where $f\colon \mathbb{N} \to \mathbb{N}$ is a computable function.
\end{enumerate}
The output $\mathcal{A}(x,k)$ is called a \emph{kernel}.
If $f$ is polynomial, then the kernel is called a \emph{polynomial kernel}.

We employ the framework of OR-cross-compositions~\cite{kernel:BodlaenderJK14}, which is a standard technique for establishing kernelization lower bounds.

\begin{definition}
An equivalence relation $R$ on $\Sigma^{*}$ is a \emph{polynomial equivalence relation} if:
\begin{enumerate}
    \item for all $x,y \in \Sigma^{*}$, whether $(x,y) \in R$ holds can be decided in time polynomial in $|x|+|y|$, and
    \item for every $n \in \mathbb{N}$, the number of equivalence classes of $R$ restricted to strings of length at most $n$ is bounded by a polynomial in $n$.
\end{enumerate}
\end{definition}

\begin{definition}
Let $L \subseteq \Sigma^{*}$ be a language,
let $R$ be a polynomial equivalence relation on $\Sigma^{*}$, and let $Q \subseteq \Sigma^{*} \times \mathbb{N}$ be a parameterized problem.
An \emph{OR-cross-composition} from $L$ into $Q$ (with respect to $R$) is an algorithm that, given $t$ instances $x_1,x_2, \ldots,x_t \in \Sigma^{*}$ that belong to the same equivalence class of $R$, runs in time polynomial in $\sum_{j=1}^{t} |x_j|$ and outputs an instance $(y,k) \in \Sigma^{*} \times \mathbb{N}$ such that:
\begin{enumerate}
    \item the parameter $k$ is bounded by a polynomial in
    $\max_{j \in [t]} |x_j| + \log t$, and
    \item $(y,k) \in Q$ if and only if $x_j \in L$ for some $j \in [t]$.
\end{enumerate}
\end{definition}

We say that $L$ \emph{cross-composes} into $Q$ if there exists an OR-cross-composition from $L$ into $Q$.
The resulting instance $(y,k)$ is called the \emph{composed instance}.

\begin{theorem}[{\cite{book:CyganFKLMPPS15}}]
    \label{thm:cross_compose}
    If an $\NP$-hard language $L$ cross-composes into a parameterized problem $Q$, then $Q$ does not admit a polynomial kernel unless $\coNP \subseteq \NP/\mathrm{poly}$.
\end{theorem}

The main result of this section is as follows.

\vcnopolykernel*\label{thm:vc_no_poly_kernel*}

Our proof follows the standard approach used to derive kernelization lower bounds for \prb{Dominating Set} parameterized by the vertex cover number~\cite{book:CyganFKLMPPS15,DS:DomLS14}.
The approach employs \prb{Red-Blue Dominating Set} as an intermediate problem and derives a polynomial-parameter transformation from it.
We adopt the same strategy and introduce the following red-blue variant of \prbOD.

The \textsc{Colored Red-Blue Odd Domination (Col-RBOD)} problem is defined as follows.
The input consists of a bipartite graph $G=(R \cup B,E)$ with bipartition $(R,B)$,
an integer $k$, and a partition of $R$ into $k$ sets $R^1,R^2,\dots,R^k$.
The question is whether there exists a set $X \subseteq R$ such that $|X|=k$, $|X \cap R^i|=1$ for every $i \in [k]$, and every vertex in $B$ has an odd number of neighbors in $X$.
For convenience, we abuse the notation for \prbGen as the one for \prb{Col-RBOD}.

We begin by observing that \prb{Col-RBOD} is $\NP$-hard.

\begin{lemma}
    \label{lem:RBOD_hard}
    \prb{Col-RBOD} is $\NP$-hard.
\end{lemma}

\begin{proof}
    We present a polynomial-time reduction from \prbOD.
    Let $(G,k)$ be an instance of \prbOD.
    We construct an instance $(G',k)$ of \prb{Col-RBOD} as follows.

    For each $i \in [k]$, let $V^i$ be a copy of $V(G)$, and for each $x \in V(G)$, let $x^i$ denote the copy of $x$ in $V^i$.
    We additionally introduce vertices $z^1,z^2,\ldots,z^k$.
    Let $R^i \coloneq V^i \cup \{z^i\}$, $R\coloneq \bigcup_{i=1}^{k} R^i$, and $B\coloneq V(G)$ and define a bipartite graph $G'=(R \cup B,E')$.

    We add edges only between $R$ and $B$ as follows.
    For every vertex $x \in V(G)$ and every $i \in [k]$, add the edge between $x^i \in R^i$ and $x \in B$.
    Furthermore, for every edge $xy \in E(G)$ and every $i \in [k]$, add the edge between $x^i \in R^i$ and $y \in B$.
    This completes the construction. 

    For the correctness, we first suppose that $(G,k)$ is a yes-instance of \prbOD, and let $D \subseteq V(G)$ be an ODD-set of $G$ with $|D| = k' \leq k$.
    We write $D=\{d_1,d_2,\ldots,d_{k'}\}$. 
    We now define a set $D'\subseteq R$ by
    \[
    D'\cap R^i = 
    \begin{cases}
        \{d_i^i\} & \text{if } i\le k',\\
        \{z^i\} & \text{otherwise}.
    \end{cases}
    \]
    Then $|D'|=k$ and $|D'\cap R^i|=1$ for all $i\in [k]$.
    Since each vertex $z^i$ has no neighbors in $B$, for every $v\in B=V(G)$, we obtain 
    \[
    |N_{G'}(v)\cap D'|
    =
    |\{d_i^i \mid d_i\in N_G[v]\}|
    =
    |N_G[v]\cap D|.
    \]
    \revn{Since $D$ is an ODD-set of $G$, the right-hand side is odd.}
    Hence every vertex of $B$ has an odd number of neighbors in $D'$,
    and $(G',k)$ is a yes-instance of \prb{Col-RBOD}.

    Conversely, we suppose that $(G',k)$ is a yes-instance of \prb{Col-RBOD}, and let $D' \subseteq R$ be a set satisfying $|D' \cap R^i| = 1$ for every $i \in [k]$ and the parity conditions for $B$.
    For each vertex $u \in V(G)$, define
    \[
    \mu(u) \coloneq |\{ i\in [k] \mid u^i\in D' \}|.
    \]
    Let $D \coloneq \{u\in V(G) \mid \mu(u)\equiv_2 1\}$.
    Since $|D'| = k$, we have $|D|\le k$.

    Fix an arbitrary vertex $v \in V(G)=B$.
    By construction of $G'$, we have 
    \[
        |N_{G'}(v)\cap D'| = \sum_{u\in N_G[v]} \mu(u) \equiv_2 |N_G[v]\cap D|.
    \]
    Since $D'$ satisfies the parity condition in $G'$, the left-hand side is odd.
    Therefore, $|N_G[v]\cap D|$ is odd.
    As $v$ was arbitrary, every vertex of $G$ is oddly dominated by $D$.
    Consequently, $D$ is an ODD-set of $G$ of size at most $k$, implying that $(G,k)$ is a yes-instance.
\end{proof}

\begin{lemma}
    \label{lem:RBOD_lowerbound}
    \prb{Col-RBOD} does not admit a polynomial kernel when parameterized by $|B| + k$, unless $\coNP \subseteq \NP/\textnormal{poly}$.
\end{lemma}

\begin{proof}
    We prove the claim by giving an OR-cross-composition of \prb{Col-RBOD} into itself parameterized by $|B|+k$.
    
    We first define a polynomial equivalence relation $R$ as follows. 
    Two instances $(G_1 = (R_1 \cup B_1, E_1), k_1)$ and $(G_2 = (R_2 \cup B_2, E_2), k_2)$ of \prb{Col-RBOD} are equivalent under $R$ if and only if $|B_1| = |B_2|$ and \revn{$k_1=k_2$}.
    \revn{
    This relation is clearly decidable in polynomial time.
    Moreover, among instances of size at most $n$, both $|B|$ and $k$ are bounded by $n$, and hence the number of equivalence classes is polynomial in $n$.
    }

    \rev{
    In what follows, we construct an \emph{instance selector} \cite[Chapter~18]{book:fomin2019kernelization} for \prb{Col-RBOD}, that is, a polynomial-time algorithm that combines two equivalent instances into a single one.
    }
    \rev{Formally,} given two equivalent instances $I_1=(G_1=(R_1\cup B_1,E_1),k)$ and $I_2=(G_2=(R_2\cup B_2,E_2),k)$, the algorithm $\rho$ outputs an instance $I^*=(G^*=(R^*\cup B^*,E^*),k^*)$ satisfying the following properties:
    \begin{enumerate}
        \item $I^*$ is a yes-instance if and only if at least one of $I_1$ and $I_2$ is a yes-instance;
        \item $\rho$ runs in time polynomial in $|I_1|+|I_2|$;
        \item $|B^*|+k^* \leq c_1 \cdot (|B_1|+k) + c_2$ for some constants $c_1, c_2$.
    \end{enumerate}
    The following construction serves as an instance selector for \prb{Col-RBOD}. 
    By repeatedly applying it, we can combine an arbitrary number of equivalent instances into a single instance while increasing the parameter by only a logarithmic factor.

    Let $p \coloneq |B_1|=|B_2|$. 
    For $i\in\{1,2\}$, let $R_i^1,R_i^2,\ldots,R_i^k$ be the prescribed partition of $R_i$.
    
    We construct $G^*$ as follows.
    Let $R^* \coloneq R_1 \cup R_2$, and let $R_1^1\cup R_2^1, R_1^2\cup R_2^2, \ldots, R_1^k\cup R_2^k$ be a partition of $R^*$. 
    Identify $B_1$ and $B_2$ by an arbitrary bijection and denote the resulting set by $B'$.
    Introduce two additional sets $U=\{u_1,u_2,\dots,u_{k-1}\}$ and $V=\{v_1,v_2,\dots,v_{k-1}\}$, and define $B^* \coloneq B' \cup U \cup V$.
    Finally, set $k^* \coloneq k$. 

    For each $i\in\{1,2\}$, add all edges between $R_i$ and $B'$ according to the edges of $E_i$.
    Moreover, for every $r\in[k-1]$, add the following edges:
    \begin{itemize}
        \item connect $u_r$ to every vertex of
        $R_1^r \cup R_2^k$;
        \item connect $v_r$ to every vertex of
        $R_2^r \cup R_1^k$.
    \end{itemize}
    This completes the construction of the graph $G^*$.

    Since $|B^*|=p+2(k-1)$, we obtain $|B^*|+k^* \leq p + 3k - 2$.
    \revn{Hence Condition~(3) holds.}
    \revn{The construction is clearly computable in polynomial time, and hence Condition~(2) follows.}

    \revn{We now prove the condition~(1).}
    Suppose that $(G^*,k)$ is a yes-instance, and let $D\subseteq R^*$ be a set such that $|D|=k$, $|D \cap (R_1^i \cup R_2^i)|=1$ for every $i \in [k]$,
    and every vertex in $B^*$ has an odd number of neighbors in $D$.
    \rev{Consider an integer $r\in \{1,2\}$ such that $|D\cap R_r^k|=1$: we may assume that $r=1$ by symmetry, and hence $D\cap R_2^k=\emptyset$.}

    We claim that $D\subseteq R_1$.
    Assume for contradiction that $D$ contains a vertex of $R_2$, \rev{specifically a vertex in} $R_2^r$ with $r\in[k-1]$.
    Since $|D\cap (R_1^r\cup R_2^r)|=1$, we obtain $D\cap R_1^r=\emptyset$.
    By construction, $N_{G^*}(u_r)=R_1^r \cup R_2^k$.
    Since $D\cap R_1^r=\emptyset$ and $D\cap R_2^k=\emptyset$, it follows that $|N_{G^*}(u_r)\cap D|=0$, contradicting the parity condition at $u_r$ with respect to $D$.
    Hence $D\subseteq R_1$.

    Let $D_1 \coloneq D\cap R_1$.
    Then $|D_1\cap R_1^r|=1$ for each $r\in[k]$.
	Furthermore, for every vertex $v\in B'$, its neighborhood inside $R_1$ is exactly its neighborhood in $G_1$.
	Since $D$ contains no vertex of $R_2$, the parity condition at $v$ in $G^*$ coincides with the parity condition in $G_1$.
    Therefore, $(G_1,k)$ is a yes-instance.

    Conversely, suppose that at least one of $I_1$ and $I_2$ is a yes-instance.
    Without loss of generality, assume that $(G_1,k)$ is a yes-instance with solution $D_1\subseteq R_1$.
    Let $D \coloneq D_1$.
    Then $|D\cap (R_1^r\cup R_2^r)|=1$ for all $r \in [k]$.
	\revn{The parity conditions for vertices in $B'$ are satisfied because $D_1$ is a solution of $G_1$.}
    
    For every $r\in[k-1]$, we have $N_{G^*}(u_r)=R_1^r \cup R_2^k$ and $N_{G^*}(v_r)=R_2^r \cup R_1^k$.
    Since $D\subseteq R_1$, 
    \[
    |N_{G^*}(u_r)\cap D|=|D\cap R_1^r|=1
    \quad \text{and} \quad
    |N_{G^*}(v_r)\cap D|=|D\cap R_1^k|=1.
    \]
	\revn{Thus every vertex of $U\cup V$ also satisfies the parity condition.}
	Hence $(G^*,k)$ is a yes-instance.

    \revn{
    Given $t$ equivalent instances, we apply $\rho$ repeatedly in a binary-tree fashion.
    After at most $\ell = \lceil \log t\rceil$ rounds, we obtain a single instance $(G^*=(R^*\cup B^*,E^*), k^*)$ that is a yes-instance if and only if at least one of the input instances is a yes-instance.
    Since each application of $\rho$ increases the parameter by at most a constant, we obtain $|B^*| \leq p + 2\ell(k-1)$ and hence $|B^*| + k^* \leq p + 2\ell(k-1) + k = O(p+k\log t)$.
    Hence this yields an OR-cross-composition of \prb{Col-RBOD} into itself parameterized by $|B|+k$.
    The \Cref{thm:cross_compose} therefore implies that \prb{Col-RBOD} does not admit a polynomial kernel unless
	$\coNP \subseteq \NP/\textnormal{poly}$.
    }
\end{proof}

To prove \Cref{thm:vc_no_poly_kernel}, we give polynomial-parameter transformations from Col-RBOD to \prbGen\ and \prbOD.
The reduction preserves equivalence of instances and increases the parameter only polynomially.
Together with \Cref{lem:RBOD_lowerbound}, this implies the claimed kernelization lower bound.

\medskip
\noindent
\textbf{Proof of \Cref{thm:vc_no_poly_kernel}.}
Let $(G=(R\cup B,E), k)$ be an instance of \prb{Col-RBOD}, where $R$ is partitioned into $R^1, R^2, \ldots, R^k$.

\paragraph{\prbGen\ parameterized by vertex cover number.}
We construct an instance $(G',\Vopen,\Vclose,f,k')$ of \prbGen\ as follows.
For each $i\in[k]$, introduce a new vertex $z_i$ adjacent to every vertex of $R^i$.
Let $V(G') \coloneq V(G) \cup \{z_1, z_2, \ldots, z_k\}$, $\Vopen \coloneq V(G')$, $\Vclose \coloneq \emptyset$, and $k' \coloneq k$.
Define the offset function $f$ by setting $f(v)=0$ for all $v\in R$ and $f(v)=1$ for all $v\in V(G')\setminus R$.

We claim that $(G,k)$ is a yes-instance of \prb{Col-RBOD} if and only if $(G',\Vopen,\Vclose,f,k')$ is a yes-instance of \prbGen.

Suppose that $(G,k)$ is a yes-instance of \prb{Col-RBOD}, and let $D \subseteq R$ be a set satisfying $|D \cap R^i| = 1$ for every $i\in[k]$ and the parity conditions for all vertices of $B$.
Set $D'\coloneq D$.
Since the parity constraints on vertices of $B$ are unchanged, every vertex of $B$ satisfies its parity condition in $G'$.
For each $i\in[k]$, we have $|N_{G'}(z_i)\cap D'| = |D'\cap R^i| = 1$.
Since $f(z_i)=1$, it follows that $|N_{G'}(z_i)\cap D'|+f(z_i)\equiv_2 0$.
Hence every vertex $z_i$ also satisfies its parity condition, and therefore $(G',\Vopen,\Vclose,f,k')$ is a yes-instance of \prbGen.

Conversely, let $D'\subseteq V(G')$ be a PD-set with $|D'|\leq k$.
For each $i\in[k]$, since $N_{G'}(z_i)=R^i$ and $f(z_i)=1$, the parity constraint at $z_i$ implies $|D'\cap R^i| = |N_{G'}(z_i)\cap D'| \equiv_2 1$.
In particular, $D'\cap R^i\neq\emptyset$ for every $i\in[k]$.
Since there are $k$ sets $R^i$ and $|D'|\le k$, it follows that $|D'\cap R^i|=1$ for every $i\in[k]$ and $D'\subseteq R$.
The remaining parity constraints are exactly those of the original instance $(G,k)$.
Therefore, $(G,k)$ is a yes-instance of \prb{Col-RBOD}.

Finally, observe that $G'$ is bipartite with bipartition $(R,\, B\cup \{z_1,z_2,\ldots,z_k\})$.
Hence $B\cup \{z_1,z_2,\ldots,z_k\}$ is a vertex cover of size $|B|+k$.
Therefore, this is a polynomial-parameter transformation from \prb{Col-RBOD} to \prbGen\ parameterized by vertex cover number.

\paragraph{\prbOD\ parameterized by cluster deletion number.}
We construct an instance $(G',k')$ of \prbOD\ as follows.
For each $i\in[k]$, we introduce two new vertices $z_i$ and $z'_i$ to $V(G')$.
Then make both $R^i\cup\{z_i\}$ and $R^i\cup\{z'_i\}$ cliques.
Set $k'\coloneq k$.

We claim that $(G,k)$ is a yes-instance of \prb{Col-RBOD} if and only if $(G',k')$ is a yes-instance of \prbOD.
The only-if direction follows from an argument analogous to the one in the previous section.

Conversely, let $D'\subseteq V(G')$ be an ODD-set with $|D'|\le k$.
Since both $z_i$ and $z'_i$ are adjacent exactly to the vertices of $R^i$, their parity constraints imply that $D' \cap (R^i \cup \{z_i, z'_i\}) \neq \emptyset$ for every $i \in [k]$.
In fact, from $|D'|\leq k$, we obtain $|D' \cap (R^i \cup \{z_i, z'_i\})| = 1$ for every $i \in [k]$.

Moreover, we observe that $|D' \cap R^i| = 1$; otherwise, both $z_i$ and $z'_i$ would be contained in $D'$, a contradiction.
Since the parity constraints on vertices of $B$ are unchanged, the set $D'\cap R$ is a solution to $(G,k)$.
Therefore, $(G,k)$ is a yes-instance of \prb{Col-RBOD}.

Deleting $B \cup \{z_i, z'_i \mid i \in [k]\}$ from $G'$ leaves the disjoint union of cliques $R^1,\dots,R^k$.
Hence $G'$ has a cluster deletion set of size $|B| + 2k$.
Therefore, we obtain a polynomial-parameter transformation from \prb{Col-RBOD} to \prbOD\ parameterized by cluster deletion number.

\paragraph{{\prbOD} parameterized by feedback vertex set number and tree-depth.}

\revn{
First, we may assume without loss of generality that $|R^i|>3k$ for every $i\in[k]$.
Indeed, if $|R^i|\le 3k$ for some $i\in[k]$, we may add isolated vertices to $R^i$ until its size exceeds $3k$.
This operation does not affect the existence of a solution with size at most $k$.
}

We construct an instance $(G'=(V',E'),k')$ of \prbOD\ as follows.
For every vertex $r\in R$, introduce a copy $\tilde{r}$ of $r$.
For each $i\in[k]$, let $\tilde{R}^i \coloneq \{\tilde{r}\mid r\in R^i\}$ and $\tilde{R} \coloneq \bigcup_{i\in [k]} \tilde{R}^i$.
Furthermore, for every $i\in[k]$, introduce vertices $x^i$, $y^i$, and $z_1^i, z_2^i,\ldots,z_{3k+1}^i$.
Define $X\coloneq \{x^i\mid i\in [k]\}$, $Y\coloneq \{y^i\mid i\in [k]\}$, $Z^i\coloneq \{z_j^i \mid j\in [3k+1]\}$, and $Z\coloneq \bigcup_{i\in [k]} Z^i$.
The vertex set of $G'$ is $V'\coloneq X\cup Y\cup Z\cup R\cup \tilde{R} \cup B$.

Starting from the edges of $G$, we add the following edges.
For every $i\in[k]$, every $r\in R^i$, and every $j\in[3k+1]$, add the edges $z_j^i r$, $x^i r$, $x^i \tilde{r}$, and the edge $x^iy^i$.
Finally, for every $r\in R$, add the edge $r\tilde{r}$.


Set $k'\coloneq 3k$.
This completes the construction.

The proof of the following claim is deferred.

\begin{claim}
    \label{clm:fvs_td_bound}
    The feedback vertex set number and the tree-depth of $G'$ are bounded by a polynomial in $|B|+k$.
\end{claim}

We now prove correctness.
Suppose first that $(G,k)$ is a yes-instance of \prb{Col-RBOD}.
Let $D = \{r^1, r^2, \dots, r^k\} \subseteq R$ be a solution, where $D\cap R^i=\{r^i\}$ for $i\in [k]$.
Define $D'\coloneq \{r^i,\tilde{r}^i,x^i\mid i\in [k]\}$.
Clearly, $|D'|=3k$.

Fix an integer $i\in [k]$.
We have $N_{G'}[x^i]\cap D'=\{r^i, \tilde{r}^i, x^i\}$ and $N_{G'}[y^i]\cap D'=\{x^i\}$, and for every $j\in[3k+1]$, $N_{G'}[z_j^i]\cap D'=\{r^i\}$.
Moreover, 
\[
N_{G'}[r^i]\cap D'=N_{G'}[\tilde{r}^i]\cap D'=\{r^i, \tilde{r}^i, x^i\},
\]
while for every $r\in R^i\setminus\{r^i\}$,
\[
N_{G'}[r]\cap D'=N_{G'}[\tilde{r}]\cap D'=\{x^i\}.
\]
Hence every vertex outside $B$ satisfies the parity condition.

Since the edges incident with vertices of $B$ are exactly those inherited from $G$, every vertex of $B$ satisfies the parity condition with respect to $D'$.
Moreover, since $D$ is an ODD-set of $G$, all vertices of $D$ satisfy the parity condition in $G'$ as well.
Therefore, $D'$ is an ODD-set of $G'$ of size $3k$.



Conversely, let $D'$ be an ODD-set of $G'$ with $|D'|\le 3k$.
We first show that $R^i\cap D'\neq\emptyset$ for every $i\in[k]$. 
Suppose otherwise, and fix $i\in[k]$ such that $R^i\cap D'=\emptyset$.
Since $N_{G'}[z_j^i]=R^i\cup\{z_j^i\}$, the parity condition forces $z_j^i\in D'$ for every $j\in[3k+1]$.
Hence $Z^i\subseteq D'$, implying $|D'|\ge 3k+1$, a contradiction.


Next, we show that $\tilde{r}^i\cap D'\neq\emptyset$ for every $i\in[k]$. 
Suppose otherwise, and fix $i\in[k]$ such that $\tilde{R}^i\cap D'=\emptyset$.
By the previous paragraph, there exists a vertex $r\in R^i\cap D'$.
Since $|R^i|>3k$ and $|D'|\le 3k$, there also exists a vertex $r'\in R^i\setminus D'$.
Observe that  $N_{G'}[\tilde{r}] = \{ r, \tilde{r}, x^i \}$ and $N_{G'}[\tilde{r}'] = \{ r', \tilde{r}', x^i \}$.
Since $r\in D'$ while $r',\tilde{r},\tilde{r}'\notin D'$, 
\revn{the sets $N_{G'}[\tilde{r}]\cap D'$ and $N_{G'}[\tilde{r}']\cap D'$ differ by exactly one vertex, regardless of whether $x^i\in D'$.}
\revn{Hence $|N_{G'}[\tilde{r}]\cap D'| \not\equiv_2 |N_{G'}[\tilde{r}']\cap D'|$, contradicting the assumption that $D'$ is an ODD-set.}


Furthermore, since $N_{G'}[y^i]=\{x^i,y^i\}$, the parity condition implies $\{x^i, y^i\} \cap D' \neq \emptyset$ for every $i \in [k]$. 
In summary, for each $i\in [k]$, it holds that $R^i\cap D' \neq\emptyset$, $\tilde{R}^i \cap D' \neq \emptyset$, and $\{x^i, y^i\} \cap D' \neq \emptyset$.
Since the three sets $R^i$, $\tilde{R}^i$, and $\{x^i, y^i\}$ are pairwise disjoint and $|D'| \le 3k$, it follows that $|R^i \cap D'| = 1$ for each $i \in [k]$.

Now, let $D \coloneq D' \cap R$.
Since the edges incident with vertices of $B$ are exactly those inherited from $G$, the set $D$ satisfies the parity condition for every vertex of $B$.
Moreover, we have $|R^i \cap D| = |R^i \cap D'| = 1$ for every $i \in [k]$.
Therefore, $D$ is an ODD-set of $G$ and $(G,k)$ is a yes-instance.

\begin{proof}[Proof of the claim]
    Consider the induced subgraph $G[R\cup \tilde{R}]$.
    Since this graph is a disjoint union of copies of $K_2$, it is a forest.
    Hence, its feedback vertex set number is $0$, and its tree-depth is $2$.
    The remaining vertices belong to $X \cup Y \cup Z \cup B$, whose size is $|X|+|Y|+|Z|+|B|=k+ k+ k(3k+1)+|B| = 3k^2 + 3k + |B|$.
    
    \revn{Deleting all vertices of $X\cup Y\cup Z\cup B$ leaves the forest $G'[R\cup\tilde{R}]$.
    Therefore, the feedback vertex number of $G'$ is at most $3k^2+3k+|B|$.}

    \revn{Furthermore, placing all vertices of $X\cup Y\cup Z\cup B$ above an elimination forest of $G'[R\cup\tilde{R}]$ shows that tree-depth of $G'$ is at most $3k^2+3k+|B|+2$.}
    
    Hence both the feedback vertex set number and the tree-depth of $G'$ are bounded by a polynomial in $|B|+k$.
\end{proof}












\end{document}